\newif\ifdraft\draftfalse
\newif\ifarxiv\arxivfalse

\arxivtrue

\ifarxiv
\documentclass{article}
\usepackage[square, numbers, sort]{natbib}
\usepackage[dvipsnames]{xcolor}
\else
\documentclass[acmsmall]{ec19acm}

\AtBeginDocument{%
	\providecommand\BibTeX{{%
			\normalfont B\kern-0.5em{\scshape i\kern-0.25em b}\kern-0.8em\TeX}}}

\copyrightyear{2019} 
\acmYear{2019} 
\setcopyright{acmcopyright}
\acmConference[EC '19]{ACM EC '19: ACM Conference on Economics and Computation}{June 24--28, 2019}{Phoenix, AZ, USA}
\acmBooktitle{ACM EC '19: ACM Conference on Economics and Computation (EC '19), June 24--28, 2019, Phoenix, AZ, USA}
\acmPrice{15.00}
\acmDOI{10.1145/3328526.3329607}
\acmISBN{978-1-4503-6792-9/19/06}

\fi

\usepackage{comment}

\usepackage{fullpage}
\usepackage{amsmath, amsthm, amsfonts}

\usepackage{cleveref}
\usepackage{mathtools}
\usepackage{url}
\usepackage{soul}

\ifarxiv
\makeatletter
\def\thmhead@plain#1#2#3{%
	\thmname{#1}\thmnumber{\@ifnotempty{#1}{ }\@upn{#2}}%
	\thmnote{ {\the\thm@notefont#3}}}
\let\thmhead\thmhead@plain
\fi

\makeatother
\newtheorem{claim}{Claim}
\newtheorem{fact}{Fact}
\newtheorem{lemma}{Lemma}
\newtheorem{theorem}{Theorem}
\newtheorem{corollary}{Corollary}
\newtheorem{definition}{Definition}

\newtheorem{proposition}{Proposition}

\newtheorem*{claim*}{Claim}
\newtheorem*{theorem*}{Theorem}
\newtheorem*{lemma*}{Lemma}
\newtheorem*{corollary*}{Corollary}
\newtheorem*{proposition*}{Proposition}
\newtheorem*{remark*}{Remark}

\newcommand{\E}{\mathbb{E}}

\DeclareMathOperator*{\argmax}{arg\,max}
\newcommand{\bigo}{\mathcal{O}}
\DeclareMathOperator{\Var}{Var}
\newcommand{\info}{\mathcal{I}}
\newcommand{\PoP}{\mathrm{PoP}}
\newcommand{\noisy}{\tilde{\theta}}
\newcommand{\arxivurl}{\url{https://arxiv.org/abs/xxxx.xxxxxx}}
\newcommand{\mute}[1]{}
\newtheorem*{definition*}{Definition}

\begin{document}
\title{Price of Privacy in the Keynesian Beauty Contest}

\ifarxiv
\author{
	Hadi Elzayn\thanks{Applied Mathematics and Computational Science, University of Pennsylvania. Email: \texttt{hads@sas.upenn.edu}.}
	\and
	Zachary Schutzman\thanks{Department of Computer and Information Science, University of Pennsylvania. Email: \texttt{ianzach@seas.upenn.edu}.}
}

\maketitle
\begin{abstract}
	The Keynesian Beauty Contest is a classical game in which strategic agents seek to both accurately guess the true state of the world as well as the average action of all agents. We study an augmentation of this game where agents are concerned about revealing their private information and additionally suffer a loss based on how well an observer can infer their private signals.  We solve for an equilibrium of this augmented game and quantify the loss of social welfare as a result of agents acting to obscure their private information, which we call the `price of privacy'.  We analyze two versions of this this price: one from the perspective of the agents measuring their diminished ability to coordinate due to acting to obscure their information and another from the perspective of an aggregator whose statistical estimate of the true state of the world is of lower precision due to the agents adding random noise to their actions. We show that these quantities are high when agents care very strongly about protecting their personal information and low when the quality of the signals the agents receive is poor.

\end{abstract}
\fi

\section{Introduction}
In recent years, the mathematical study of {privacy} has become a major subject of inquiry. Much of the impetus for this work has been a series of data breaches and deanonymization of seemingly safe private information, perhaps most famously in the use of IMDb reviews to attack Netflix's data set over a decade ago \cite{narayanan2008robust}.  Models such as differential privacy alleviate this problem by providing formal guarantees about how much about any individual an adversary can learn from the release of some statistic computed from a dataset.
However, such techniques are generally predicated on the presence of a trusted central agent which applies the differentially private mechanism to the data it collects from individuals. Alternatively, in the case of the \textit{local} differential privacy model, the agents are typically instructed in how to privatize their data.

In this work, we endogenize a notion of privacy in the absence of a trusted party to coordinate the mechanism. In particular, we analyze a formalization of a classical game called the \textit{Keynesian Beauty Contest}, which has been used to study strategic interaction involving information acquisition and the coordination of collective action. In particular, we show how the traditional formulation of this game neglects privacy concerns in its equilibrium predictions, and we then provide a framework for extending the game to incorporate a flexible notion of privacy into the utility of the agents. Using this, we can characterize a \emph{price of privacy}, somewhat akin to quantities such as the price of anarchy, which measures the loss of social welfare in a population of agents who act selfishly to guarantee their own privacy.

Abstractly, we think about players in a game being perfectly rational Bayesian agents who observe some information, perform a utility-maximizing computation, and play an action. However, by observing some player $i$'s chosen action, player $j$ may be able to learn something about player $i$'s private information.  For this reason, if players fear that their public actions may reveal private information, they may be incentivized to deviate from the strategy which maximizes utility in the absence of privacy concerns. If all players share such concerns, or anticipate others harboring them, equilibrium behavior may be significantly different than standard predictions. 

One setting where there may be tension between privacy and coordination with other players and some underlying ground truth is intra-organizational information aggregation. Suppose that a firm wishes to poll a group of employees as to a particular decision, such as a prediction about the success of a particular product, an evaluation of a colleague, or a new procedure for evaluating and settling claims. In order to make the best decision possible, it is valuable to aggregate information, opinions, or signals received by individuals; however, individuals may be reluctant to fully share their opinion lest it be held against them or they prove to be `wrong'. In this situation, simple anonymization procedures may not be trustworthy or effective, and any such procedure requires trusting an internal mechanism to which the employee does not have access.  If the respondent is concerned about other agents deanonymizing her survey, then it should be expected that she does not respond honestly.

\subsection{Our results}

We consider a formal model of the Keynesian Beauty Contest, a game in which each agent observes some information and submits an estimate about the true `state of the world', then earns utility based on both how close her action is to the truth as well as how close it is to the \textit{average action} over all agents.  We describe the game, the information structure, and the strategy space and show the existence of a \textit{symmetric linear Nash equilibrium}, where agents' actions are a convex combination of the public and private signals they observe, extending the results in \citet{morrisshin}.

We then turn to the privacy-augmented version of the game, where agents face the same utility function but also suffer a loss of utility based on the ability of other players to infer their private information.  We show that this new game also has an equilibrium in strategies that can be written as strategies in the original game with added random noise.
This leads to two different values which can be thought of as a `price of privacy'.  In the first, we consider the perspective of the agents and quantify the total loss of the players' utilities as a result of incorporating this concern for privacy.  In the second, we think about an `untrusted aggregator' who wants to compute some statistic using the agents' private beliefs but cannot convince them to participate in a privacy-protecting mechanism. Here, we compute the decrease in the quality of the aggregation due to the players' addition of noise to their previously optimal actions.

At a high level, we a consider setting in which there is no centralized planning mechanism, either to perform the differentially private computation or to instruct the agents who own the data to add a particular amount of noise to their information in order to perform a locally differentially private aggregation.  Rather, we assume that agents are rational and derive some utility from both the aggregate-level computation being accurate as well as from the privacy gained by obscuring the information she releases. A major departure from the local differential privacy framework is in our treatment of outliers.  In local differential privacy, an individual whose data is very different from the norm may have to add a considerable amount of noise before providing her information to the aggregator, since we need to worry that her unmodified data might shift the distributions of outcomes farther than we would like.  In this work, we consider a different flavor of individual privacy, where an agent with outlying information does not necessarily care about being revealed as an outlier; she only cares how accurately an observer can guess her private information.

\subsection{Related Work}
\subsubsection{Economics}
Concretely, our work builds primarily on the results in \citet{morrisshin}, which formalizes the modern version of a Keynesian Beauty Contest in order to study the tendency for individuals to over-weight public information. We similarly explore linear equilibria in a Keynesian Beauty Contest game with public and private signals; in fact, the structure of our model is substantially very similar, and we recover their results as a special case when there is no concern for privacy. Viewing our model as a generalization of theirs, we expand their result on the existence of such equilibria to a privacy-aware setting. The results in \citet{hellwigveldkamp} build on those in \cite{morrisshin}, exploring a setting in which agents optimize a selection of information sources with different costs matching their qualities. While \cite{hellwigveldkamp} is quite different mechanically and in spirit, there is a certain sense in which our paper can be interpreted along similar lines: the equilibrium noisy  action can be viewed as similar to choosing a private signal of differing variances, and the privacy cost can be thought of as a cost to more precise information.  The Keynesian Beauty Contest is well-studied in the broader macroeconomics literature (see \citet{nagel2017inspired} for a recent survey) and is frequently used to study settings where agents derive some value for correctness and for coordination, such as in financial markets or strategic voting.  These settings are natural cases for the use of formal privacy methods, such as financial analysts wanting to conceal the model they use to predict asset prices or voters being concerned about their preferences being held against them. \citet{gradwohl2017perception} study a setting in which strategic agents are concerned that their actions in a game will reveal sensitive private information, but their model and approach, which is based on the concept of \textit{signaling games}, is considerably different from the setting we consider.

\subsubsection{Differential Privacy}

In the last decade, mathematical notions of privacy have been studied extensively in the computer science literature.  The most influential is \textit{differential privacy} (c.f. \citet{dworkdp}) which states roughly that a statistical algorithm is differentially private if the probability that the algorithm gives any particular output does not change by much when we modify one row of the database the algorithm is run on. 
Standard differential privacy is principally a \textit{central}, rather than \textit{individual} notion of privacy, as the algorithms operate under the assumption that the adversary does not have access to the raw data provided to the analyst.  The concept was introduced in \citet{kasiviswanathan2011can} and recent work has shown that this notion is a powerful framework for addressing privacy concerns (c.f. \cite{localdp,joseph2018locally,bassily2018linear}).
There is also a literature on \textit{privacy in mechanism design}.  The work in \citet{chen2016truthful} and \citet{nissim2012privacy} examines settings where agents participate in a data aggregation mechanism and earn utility which is increasing in the quality of the estimate but decreasing in the data leakage. The model in \cite{chen2016truthful} is one of \textit{truthful voting}. In this work, we study a game-theoretic setting where privacy is a concern, rather than designing a privacy-preserving mechanism. 
Our work differs from the formal study of privacy in that we use a definition of which does not (necessarily) satisfy the strictness of the various versions of \textit{differential} privacy. 
Rather, we consider the perspective of strategic agents who are concerned with other players' ability to learn `too much' about their own private information.  In this way, our work is related to the issue of \textit{response bias} on surveys, where respondents do not answer questions honestly in order to avoid revealing sensitive information.  Work examining the use of randomization to alleviate this effect goes back several decades (e.g. \citet{warner1965responsebias}).   

 { Our extended game does bear a resemblance to the \emph{local model} of differential privacy in that agents add noise to their actions, and one might intuitively map the case where agents add Gaussian noise to their actions to an instance of the local Gaussian mechanism (see, e.g. \cite{dworkdp}). However, agents in our model are \emph{not} using the Gaussian mechanism to achieve differential privacy, and such a guarantee cannot be recovered\footnote{There are other, more complicated mechanisms that can achieve differential privacy with unbounded input data. We refer readers to \cite{liu2019generalized} and \cite{wang2019collecting} for more discussion. }, as the fact that agents' data is unbounded means that this mechanism cannot give non-vacuous guarantees. To see this, recall that differential privacy requires that for any set of output of the mechanism must be approximately as likely when input is \emph{neighboring}, which typically means differing in one record. In our case, this corresponds to having a different private signal. The problem is that for any choice of finite variance, adding Gaussian noise to private signals that are sufficiently distant will produce noisy actions that are far apart, and can thus be distinguished with high confidence. 
 	
  More importantly, there is a significant conceptual difference between the models of differential privacy and our model. In differential privacy models, there is a strict separation of roles: there is an accuracy-concerned learner that wishes to perform and release the output of a query, and privacy-concerned agents that supply their data.  The addition of noise {either centrally or locally }  prevents  any other party from learning   the exact value of the private data with high confidence, and is assumed to be necessary to make agents willing to supply their data. In contrast, there is no external learner in our model; instead, agents attempt to coordinate their actions with other agents and the state of the world, but must add noise to their actions in order to prevent others from learning their private information. Noisier actions lead to greater variance around the state of the world and the average actions, and thus lower payoffs. Hence, the agents care about \emph{both} privacy \emph{and} accuracy, and this tradeoff determines how much noise they ultimately add. 
 	
 }

\section{Framework and Model}
\label{sec:framework}
\subsection{The Keynesian Beauty Contest}
The Keynesian Beauty Contest dates back to John Maynard Keynes' 1936 \textit{General Theory} \cite{keynes1936general} in which he formulates a game by analogy to a newspaper beauty contest. In this game, players select the `most beautiful' entrants from an array of photographs printed in the newspaper, and those players who choose the \textit{most popular} faces are eligible for a prize.  
The salient idea is that a rational contestant must consider her own opinions about which entrants are the most beautiful as well as her \textit{beliefs} about the opinions of all other contestants, and perhaps her beliefs about the beliefs of the other players' opinions of other players, and her beliefs about the beliefs about the beliefs about the beliefs about the opinions of other players, and so on.  
Keynes originally proposed this as a model to explain the behavior of financial markets, where the value of an asset depends as much on its fundamental potential for returns as it does the \textit{collective belief in its potential} for returns. 
Since its inception, this model has been used throughout the economics literature to describe the behavior of rational agents in strategic environments (c.f. \cite{gao2008keynesian,allen2006beauty,bosch2002one,cespa2015beauty,nagel2017inspired,morrisshin}), particularly in macroeconomic settings.

In this paper, we work with a common abstraction of the Keynesian Beauty Contest (see e.g. \cite{morrisshin,hellwigveldkamp}) in which there is some true state of the world $s$, and each agent submits a guess $\theta_i$ and earns utility equal to 
$$u_i(\theta_i,\theta_{-i}) = -(1-\alpha)(\theta_i-\bar{\theta})^2 - \alpha(\theta_i-s)^2,$$
where $\bar{\theta}$ is the average choice of $\theta_i$ over all players and $\theta_{-i}$ denotes the action of all other players other than $i$.  
We refer to the $(\theta_i-\bar{\theta})^2$ term as the \textit{coordination component}, which rewards how close player $i$'s action is to that of the other players, and the $(\theta_i-s)^2$ term as the \textit{guessing component}, which describes how close player $i$'s action is to being the correct guess for the true state of the world $s$. 
The parameter $\alpha\in[0,1]$ is common to all agents and describes the relative weighting of the values of coordination and guessing the true state of the world correctly. 

We let $n$ denote the number of players.  In the economics literature, it is standard to consider a continuum of agents indexed by the unit interval $[0,1]$.  We present our results from the perspective of the finite game and observe that the analogous result for the continuum case emerge by taking the limit as $n$ grows to infinity.  We state these as propositions immediately following the corresponding result in the finite game.  For the sake of brevity, all omitted proofs as well as direct proofs of the results in the infinite game can be found in \ifarxiv the appendix.\fi

The fundamental difference between the two cases is in the definition of the average action $\bar{\theta}$.  In the game with finitely many and infinitely many agents, these are

\begin{align*}
\bar{\theta} = \frac{1}{n}\sum\limits_{j=1}^n \theta_j \qquad \text { and }\qquad \bar{\theta} = \int_{0}^{1} \theta_j dj,
\end{align*}

respectively. In the finite game, agent $i$'s action $\theta_i$ has a measurable effect on the average action whereas in the infinite game, her impact is infinitesimal.  Therefore, when there are finitely many players, agent $i$ must consider her own action's effect on $\bar{\theta}$;  in the infinite case, she can treat $\bar{\theta}$ as simply the average action of all players other than herself.  Intuitively, as  $n$ grows, any individual's ability to unilaterally affect $\bar{\theta}$, so we should expect that the results in the infinite game emerge as the limits of the results in the finite game as $n$ grows to infinity.

\subsection{Information Structure}

We now describe the information structure of the game.  
Each player has a common \textit{improper uniform} prior distribution over the value of $s$. There is a public signal $y$, drawn from a Gaussian distribution with mean $s$ and variance $\sigma^2_y$.  
Additionally, each player $i$ observes a private signal $x_i$ independently drawn from a Gaussian distribution with mean $s$ and variance $\sigma^2_x$. 
Each player observes $y$ and her own $x_i$, but not any other player $j$'s private signal $x_j$.  
Each player $i$'s utility function $u_i(\theta_i,\theta_{-i})$ is identical as given above.  The values of $\sigma^2_y$, $\sigma^2_x$, $\alpha$, the utility functions, and the realization of $y$ are public and common knowledge, and the form of the priors and structure of the game is also common knowledge.  Furthermore, we assume that the agents are rational and Bayesian.  That is, they seek to maximize their expected utility given their knowledge and their beliefs about other players.

{ We use the notion of an \textit{improper uniform prior} belief to align with the economics literature; however, this concept is not necessarily standard across disciplines. One can think of a (proper) uniform prior with compact support as indicating that players believe any underlying parameter in the support is equally likely to be the truth. An `improper uniform prior' expresses the same sentiment for an unbounded support. The term `improper' arises because such an object is not a probability density function, but, when updated with the proper Gaussian signals in our setting, the posterior is a proper distribution. \Citet{berger2009formal} provide a useful discussion of settings in which the improper uniform prior is appropriate. 

For readers uncomfortable with such a notion, we present an alternative framing of the information structure which is equivalent to the above.  The private signals $x_i$, the utility functions, and the values of $\alpha$ and $\sigma^2_x$ are as before.  Rather than $y$ representing a public signal, each agent instead has a prior belief about the true state of the world $s$, which is a Gaussian distribution with mean $y$ and variance $\sigma^2_y$.  This prior is common to all agents and the parameters of the prior are common knowledge.  This framing is equivalent to the previous one since an agent with an improper uniform prior, upon observing the realization of the public signal $y$, updates her posterior belief about $s$ to be Gaussian with mean $y$ and variance $\sigma^2_y$, which is what she would have believed if she instead began with the common Gaussian prior.}

We write $\mathcal{I}_i$ for the information set of player $i$, which captures the structural information about the game as well as the values of $x_i$ and $y$.  
Furthermore, we use $\mathbb{E}_i\left(\cdot\vert \mathcal{I}_i\right)$ to denote the \textit{expectation of player $i$ at $\mathcal{I}_i$}, which is player $i$'s belief about some quantity given that she knows everything at $\mathcal{I}_i$.  Where it is clear from context, we drop the $\mathcal{I}_i$ for notational clarity and simply write $\mathbb{E}_i(\cdot)$.

We now consider an equilibrium concept for this game. 
\begin{definition}[Symmetric linear Nash equilibrium]
	A \textit{symmetric linear Nash equilibrium} of this game is an action $\theta_i$ for each player $i$ which can be written as a linear combination of the private and public signals $x_i$ and $y$ such that no player can profitably deviate unilaterally and every player $i$ chooses the same weight to put on $x_i$ and $y$.  That is, a Nash equilibrium in which $\theta_i=\kappa x_i + (1-\kappa)y$ for all players and $\kappa$ is identical for all players.
\end{definition}

The authors in \cite{morrisshin} describe the \textit{unique} {symmetric linear Nash equilibrium} of a slightly different version of the game which has the same first order condition; our proof follows a very similar structure.

\begin{lemma}[First order condition]\label{lem:finite-foc}
	In equilibrium, an agent's optimal choice of $\theta_i^*$ must satisfy
	\begin{align*}
	\theta_i^* = \frac{\alpha n^2 \E_i[s]}{\alpha (2n-1)+(n-1)^2} + \frac{(1-\alpha)(n-1)\E_i\left[\sum\limits_{j\neq i}\theta_{j}\right]}{\alpha (2n-1)+(n-1)^2}
	\end{align*}

\end{lemma}

\begin{proposition}\label{lem:FOC}
	In the game with infinitely many agents, in equilibrium, an agent's optimal choice of $\theta_i^*$ must satisfy
	\begin{align*}
	\theta_i^* = \alpha \E_i[s]+(1-\alpha) \E_i[\bar{\theta}]
	\end{align*}
\end{proposition}
\mute{\begin{proof}
Since agent $i$ chooses his action after observing signals $x_i$ and $y$, they optimize their action given those signals. We can write:
\begin{align*}
\theta_i^* \in \argmax \E_i [u_i(\theta_i)]
\end{align*}
Expanding, we have:
\begin{align*}
\E_i [u_i(\theta_i,\theta_{-i})] =-\alpha(\E_i[\theta_i^2 -2 \theta_i s+ s^2]) -(1-\alpha)(\E_i[\theta_i^2 -2\theta_i \bar{\theta} +\bar{\theta}^2]).
\end{align*}
We can differentiate with respect to $\theta_i^*$ and set the derivative equal to zero to find
\begin{align*}
0=\frac{\partial}{\partial \theta_i} \E_i[u_i(\theta_i,\theta_{-i})] = -\alpha (2\theta_i - 2 \E_i[s])-(1-\alpha)[2\theta_i -\bar{\theta}].
\end{align*}
Where we have used the fact that $\E_i[\theta_i^*]=\theta_i^*$, since it is agent $i$'s choice and not a random variable, as well as the fact that $\frac{\partial\bar{\theta}}{\partial\theta_i^*}=0$ because there are infinitely many agents, so agent $i$'s action cannot affect the average of all of the actions.

Rearranging yields the claim.
\end{proof}}

{These results say that the optimal action is a convex combination of agents' expectations about the state and the average actions; this follows directly from the structure of the payoff function. Now, if other agents' strategies are convex combinations of signals, then the unique best response is also a convex combination of signals. We can match coefficients to solve for this symmetric linear Nash equilibrium.}

\begin{lemma}[Nash equilibrium and Value of $\kappa$]  \label{lem:finitekappa}
	The symmetric linear Nash equilibrium is given by

\begin{align*}
\theta_i &= \kappa x_i + (1-\kappa)y \text{, where} \\
	\kappa &= \frac{\alpha n^2 \tau_x}{\alpha n^2 \tau_x + \left((n-1)^2 +\alpha (2n-1)\right)\tau_y}, \ \tau_x  = \frac{1}{\sigma_x^2},\text{ and } \tau_y =\frac{1}{\sigma_y^2}
		\end{align*}
	for all players $i$.

\end{lemma}

\begin{proposition}\label{cor:infintenash}
	The symmetric linear Nash equilibrium in the game with infinitely many players is given by 

\begin{align*}
\theta_i &= \kappa x_i + (1-\kappa)y \text{, where} \\
\kappa&= \frac{\tau_x \alpha}{\tau_x \alpha +\tau_y},\ \tau_x  = \frac{1}{\sigma_x^2},\text{ and } \tau_y =\frac{1}{\sigma_y^2}
\end{align*}
for all players  $i$.
This recovers the analogous result in \cite{morrisshin}.
\end{proposition}
Notice that if $\alpha=1$, the agent is only rewarded based on his closeness to the state, and $\theta_i^*$ then dictates that agent $i$ play his best guess of the state. As $\alpha$ falls, the weight Agent $i$ places on his private signal diminishes, until at $\alpha=0$, the agent chooses $\theta_i^*$ to be equal to $y$, the public signal. 

We can now write the expected utility to each agent in this equilibrium:

\begin{lemma}[Expected utility]
	\label{lem:finite-eukbp}
	The expected utility of agents for playing $\theta_i=\kappa x_i+(1-\kappa)y$, given $s$ is:
	\begin{align*}
\E_i[u_i\vert s]= -\alpha (\kappa^2 \sigma^2_x + (1-\kappa)^2\sigma^2_y) -\frac{(1-\alpha)\kappa^2(n-1)}{n}\left.\sigma^2_x \right.
	\end{align*}
\end{lemma}

\begin{proposition}
	\label{lem:eukbp}
	The expected utility of agents for playing $\theta_i=\kappa x_i+(1-\kappa)y$, conditional on $s$ in the game with infinitely many players,  is:
	\begin{align*}
	\E_i[u_i\vert s] = -\alpha(1-\kappa)^2\sigma^2_y - \kappa^2\sigma^2_x
	\end{align*}
\end{proposition}

Using the value of $\kappa$ from the infinite setting for simplicity\footnote{For a fixed $n$, the infinite version of $\kappa$ differs from the finite version by a multiplicative factor of $1-\frac{1}{n^2}$ while the utilities differ by a factor of $1-\frac{1}{n}$.  Therefore, the value of $\kappa$ in the infinite case is very close to that in the finite case for modest values of $n$, and any countervailing effect of changing the parameters through $\kappa$ will be dominated by the direct effects on the utility.} we can now examine how utility changes as the variances of the signals and number of agents do.  We briefly analyze this next.

\begin{corollary}[Comparative statics]
	\label{cor:compstat}

	\begin{align*}
	\frac{\partial\E_i[u_i\vert s]}{\partial \sigma^2_x} = -\frac{(\alpha\sigma^2_y)^2}{(\alpha\sigma^2_y + \sigma^2_x)^3}\left( (2-\alpha)\alpha^2 \sigma^2_x \sigma^2_y + \sigma^2_y - \frac{n-1}{n}(1-\alpha)(\alpha\sigma^2_y - \sigma^2_x)   \right),
	\end{align*}

\begin{align*}
\frac{\partial\E_i[u_i\vert s]}{\partial \sigma^2_y} = -\frac{\alpha (\sigma^2_x)^2}{(\alpha\sigma^2_y + \sigma^2_x)^3} \left(2\alpha^2 \sigma^2 _y - \alpha\sigma^2_y + \sigma^2_x + \frac{n-1}{n} 2\alpha(1-\alpha)\sigma^2_y\right),
\end{align*}
and
\begin{align*}
\frac{\partial\E_i[u_i\vert s]}{\partial n}= -\frac{(1-\alpha)\alpha^2 \sigma^2_x (\sigma^2_y)^2}{n^2(\alpha\sigma^2_y + \sigma^2_x)^2}.
\end{align*}
	
\end{corollary}

All else fixed, decreasing  $\sigma^2_x$ unambiguously increases utility.  The fractional term in $\frac{\partial\E_i[u_i\vert s]}{\partial \sigma^2_x}$ is always positive and the term in the parentheses is also always positive, so $\frac{\partial\E_i[u_i\vert s]}{\partial \sigma^2_x}$ is strictly decreasing as $\sigma^2_x$ increases.  This aligns with the intuition that higher quality information means that agents will be able to both guess the true state of the world more precisely as well as coordinate better with one another.  A similar intuition holds for the effect of decreasing $\sigma^2_y$, since the term in the parentheses is always positive when $n\geq 2$.  However, the \textit{rate} at which utility increases when the signal variances decrease is not the same.  Supposing $\sigma^2_x=\sigma^2_y$, agents gain more value from decreasing $\sigma^2_y$ than $\sigma^2_x$ for two reasons.  The first is the over-weighting of $y$.  Since agents more heavily weight the public signal $y$ compared to their private information $x_i$, a small improvement in the precision of $y$ will help agents choose an action $\theta_i$ closer to $s$ than an identical improvement in the precision of the $x_i$ signals.  The second reason can be thought of as a second-order effect of this over-weighting.  Since agents over-weight it, improving the quality of the public signal $y$ even further increases the weight agents will place on it.  Therefore, not only will agents actions be closer (in expectation) to $s$, but they will be closer to \textit{each other}, thus also improving the utility in the coordination component.

As $n$ increases, utility decreases.  There is an $\tfrac{n-1}{n}$ coefficient on the second term in the utility function, and as $n$ grows, this grows towards 1.  We can think of this term as measuring the amount of non-impact any one agent can have on the average action $\bar{\theta}$.  When the number of agents is very few, agent $i$'s action $\theta_i$ can't be too far from the average since the construction of $\bar{\theta}$ will have a high weight on  $\theta_i$, and this weight decreases as $n$ grows.  For very large values of $n$, agent $i$ has a small impact on the average action, so the risk of being far away is greater.

\subsection{Privacy-Awareness}
Now we come to the heart of the privacy issue. Suppose that at the end of the game, each player's action $\theta_i$ is made public. Then, upon seeing $\theta_i$, player $j$ can simply write $$x_i = \frac{\theta_i-(1-\kappa)y}{\kappa},$$ and since everything on the right hand side is known to player $j$, she can learn player $i$'s private signal $x_i$ with perfect precision. This suggests that if the players care at all about preserving the privacy of their private signals, they ought not to play exactly the equilibrium prediction $\theta_i$. 

As written, privacy is not in the players' utility functions, and it is too much to ask for a prediction that incorporates privacy without actually incorporating privacy into the utility function. To overcome this, we let $\rho(\noisy_i)$ denote, abstractly, a measure which describes how `private' the (possibly randomized) action $\noisy_i$ is. For example, $\rho$ could represent the \textit{maximum precision (i.e. the reciprocal of the variance)\footnote{Taken as a worst-case over all $j\neq i$.} to which some player $j$ can infer the private signal $x_i$ of player $i$} after observing his action and given her information set $\mathcal{I}_j$ and knowledge of the equilibrium strategies.  Here, we consider both this precision measure of privacy as well as an \textit{entropy} measure, where $\rho(\noisy_i)$ is the (information-theoretic) entropy of this inference, rather than the variance.

It is clear that any equilibrium which prescribes a deterministic mapping from the available information to an action $\theta^*_i$ cannot offer any measure of privacy to the players.  To correct for this, we enhance players' utility functions to incorporate the value they gain from obscuring their signals and extend the equilibrium concept to a \textit{noisy} one, where players select not only a guess about the state, but also a \textit{noise-generating distribution}.
Formally, we extend the utility function to 
$$v_i(\noisy_i,\noisy_{-i}) = (1-\beta)u_i(\noisy_i,\noisy_{-i}) + \beta\rho(\noisy_i),$$ 
where $u_i$ is as before and $\beta\in[0,1]$ denotes the agents' \textit{relative value for obfuscation}.

We can observe that if $\beta\neq0$, the equilibrium actions in the original game do not support an equilibrium in this game, since upon observing $\theta^*_i$, any player $j$ can exactly recover the private signal $x_i$, and player $j$'s `distribution' over the possible values of $x_i$ is degenerate with variance or entropy equal to zero. Concretely, in the precision setting, $\beta\rho(\noisy_i)$ is $-\infty$ and $v_i(\theta_i^*,\noisy_{-i})=-\infty$. Conversely, if player $i$ were to choose $\kappa'=0$ (and therefore $\noisy_i=y$), her utility would be finite, thus demonstrating a profitable deviation.

 Thus, any deterministic equilibrium cannot be supported if privacy is incorporated into the utility function, and some degree of randomization is necessary; however, the next section will show that the prediction of the deterministic optimum will serve as the core deterministic component to an optimal randomized strategy.

We remark here that if agents do not care about others learning their signal at all (i.e. $\beta=0$), the utility function degenerates to that of the original game. If $\beta=1$, then players \textit{only} care about protecting their private signal, and a dominant strategy is to take the action which maximizes the obfuscation of the private signal.

\section{The Extended Game}
\label{sec:extend}
We now formalize this modified game and show that it still has a symmetric linear Nash equilibrium when we consider \emph{randomized} strategies. Here, optimal actions are of the form $\noisy_i = \theta_i + \eta_i$, where $\theta_i$ is the same linear aggregation of the public and private signals as in the original game and $\eta_i$ is independent noise drawn from a distribution whose form depends on how we measure the obfuscation component $\rho(\noisy_i)$.  That is to say, there is a symmetric linear Nash equilibrium in which players' optimal action is a noisy modification of their optimal action in the original game. Note that randomized strategies are \emph{not} mixed strategies in the classical sense.  Each player's choice of a noise-generating distribution can be thought of as committing to a collection of parameters which describe that distribution.  In this sense, our equilibrium concept is a `pure strategies' one, since each player will commit to a single collection of parameters, not a distribution over such collections.
\subsection{Defining the Extended Game}
In the extended game, $s$, $x_i$, $\sigma^2_x$, $y$, $\sigma^2_y$, and $\bar{\theta}$ are as before. Each player chooses an action $\noisy_i$, and receives utility 
$$v_i(\noisy_i) = (1-\beta)u_i(\noisy_i) + \beta\rho(\noisy_i),$$
 where $\beta\in[0,1]$, and $u_i(\noisy_i,\noisy_{-i}) = -(1-\alpha)(\noisy_i - \bar{\noisy})^2 - \alpha (\noisy_i - s)^2$ as before. We write $\bar{\noisy}$ to represent the average (noisy) action of the players.
After each player announces $\noisy_i$, players also learn the function each player used to select $\noisy_i$ as a function of $\mathcal{I}_i$.  For example, if a player's strategy is to play $x_i$ plus some random noise, player $j$ learns the distribution from which this noise was drawn, though (crucially) not the \textit{realization} of this draw.
We denote the \emph{privacy loss} of each player after actions are revealed $\rho(\noisy_i)$. The two possibilities we consider for $\rho$ are both functions of the belief distribution that a Bayesian agent, given their information set, observation of $\noisy_i$, and knowledge of equilibrium strategies and  the randomization mechanism used by player $i$,  assigns to $x_i$; one is the \emph{precision} of these beliefs, and the other is the \emph{entropy}. Formally,

\begin{definition}[Precision privacy loss in equilibrium]
	For a given equilibrium profile, Player $i$'s \emph{precision} privacy loss is the expected precision with which a Bayesian opponent, knowing the equilibrium profile and observing their own signal, the public signal, and Player $i$'s action $\noisy_i$, can estimate Player $i$'s signal $x_i$. That is, if $\gamma$ represents the belief distribution about $x_i$, we define:
	\begin{align*}
	 \rho_{\text{prec}}(\noisy_i) =  \frac{1}{\Var\gamma(x_i\vert s, \noisy_i,\mathcal{I}_j,H) }
	\end{align*}
\end{definition}

\begin{definition}[Entropy privacy loss in equilibrium]
	For a given equilibrium profile, Player $i$'s \emph{entropy} privacy loss s the expected precision with which a Bayesian opponent, knowing the equilibrium profile and observing their own signal, the public signal, and Player $i$'s action $\noisy_i$, can estimate Player $i$'s signal $x_i$. That is, if $\gamma$ represents the belief distribution about $x_i$, we define:
	\begin{align*}
	\rho_{\text{ent}}(\noisy_i) =  -\int \gamma(x_i\vert s, \noisy_i,\mathcal{I}_j,H) \log \gamma(x_i\vert s, \noisy_i,\mathcal{I}_j,H) d\noisy_i
	\end{align*}
\end{definition}

Before we move on, it is worth considering kinds of distributions these choices consider to be `private'. For example, considering the following two belief distributions:
\begin{align*}
\gamma_1 \ : \ x_j = U([-\epsilon,\epsilon])\qquad 
\gamma_2 \ : \ \begin{cases*}
M & with probability $\delta$ \\
-\epsilon        & otherwise
\end{cases*}
\end{align*}

By choosing $M$ large and $\epsilon$ small, the appropriate choice of $\delta$ can be made to force $\gamma_2$ to have arbitrarily large variance (small precision), but extremely low entropy. On the other hand, a uniform distribution $\gamma_1$ has relatively high entropy as compared to its variance.  If agents measure their privacy by precision, they will be indifferent between adding noise from a fairly narrow uniform distribution and a discrete distribution which almost always adds a very small amount of noise but rarely adds an enormous amount.  However, if they measure their privacy with respect to entropy, the uniform option offers a much greater amount of privacy, and since (as we will show in \Cref{lem:sep}) agents pay a price in their utility equal to the variance of the noise-generating distribution, for a fixed variance the higher entropy uniform distribution will be strictly preferred to the discrete one.

We aim to show the following theorem:
\begin{theorem}
	\label{thm:main-thm}
	Consider the modified game as defined above.  Suppose players' actions $\noisy_i$ and the (possibly randomized) mapping $(y,x_i)\mapsto \noisy_i$ are revealed at the end of the game, and players' loss of privacy $\rho(\noisy_i)$ is measured as either the reciprocal of the variance 
	 or the entropy of a representative player $j$'s posterior belief about $x_i$ at the end of the game.  Then, this game has a symmetric linear Nash equilibrium where each player chooses $\noisy_i = \kappa x_i + (1-\kappa)y + \eta_i$ where $\kappa$ is as it is in the original game and $\eta_i$ is a random variable drawn from a distribution whose form depends on whether we use reciprocal variance or entropy to measure privacy.
\end{theorem}

\begin{proof}[Proof sketch]
	We prove this theorem over the course of several steps throughout the remainder of this section, which contains the formal statements of the following:
	
	\begin{enumerate}
		\item First, we show (\Cref{clm:meanzero}) that if players do indeed choose their action by adding noise to their privacy-unaware equilibrium action, then the distribution which generates the noise must have mean zero.
		\item Using this, we prove (\Cref{lem:finite-sep}) that the utility function \textit{separates} into the sum of three parts: the utility in the original game, a penalty paid in the variance of the noise distribution, and the privacy component.
		\item Then, since the utility function separates additively, we can use the first-order conditions (\Cref{lem:finite-foc}) to solve for the optimal choice of distribution from which to draw the noise (\Cref{cor:finiteparams}), which depends on whether the loss of privacy is measured by reciprocal variance or negative entropy. 
		\item Finally, we demonstrate that this strategy profile supports a Nash equilibrium by arguing that there is no profitable deviation for any player in the game (\Cref{lem:slnne-ent,lem:slnne-var}).  We call this a \textit{symmetric noisy linear Nash equilibrium}.
	\end{enumerate}

\end{proof}

We first define a  Noisy Strategy:
\begin{definition}[(Linear) $H$-noisy strategy]

	An  $H$-noisy strategy has the form
	\begin{align*}
	\noisy(x,y,\nu) = \theta_i(x,y) + \eta_i
	\end{align*}
	where $\eta_i$ is a random variable drawn from a distribution $H$ and $\theta_i(x,y)$ is a deterministic function of signals $x$ and $y$. We say that $\noisy$ is a \emph{normal} noisy strategy if $H$ is a Gaussian distribution. If, moreover, $\theta_i(x,y)$ is \emph{linear} and can be written as $\theta_i(x,y)=  \kappa x + (1-\kappa) y$, we say that $\noisy$ is a \textit{linear} noisy strategy.
\end{definition}
Note that \emph{any} randomized strategy that can be decomposed into a deterministic component and a random component can be described as a noisy strategy. This means that Lemma \ref{lem:sep} applies whether or not the underlying deterministic component is linear.

\begin{claim}\label{clm:meanzero}
	
	If  there exists an equilibrium in noisy strategies, where player $i$ chooses $\noisy_i = \theta_i+ \eta_i$ and each player's $\eta_i$ is drawn independently from a distribution $H_i$, then there exists such an equilibrium strategy profile in which the mean of each $H_i$ is zero.
	
\end{claim}

\mute{\begin{proof}
	Since the distributions which generate the $\eta_i$ are revealed after each player announces $\noisy_i$, choosing to draw noise from a distribution with non-zero mean cannot improve the privacy that player $i$ achieves, since a representative player $j$ can simply subtract this mean when constructing her posterior distribution over $x_i$.  Additionally, choosing a mean other than zero makes the utility from the guessing component strictly worse.  Finally, if we assume that all players other than $i$ choose to add noise drawn from a distribution with the same mean, then in the coordination portion of the utility function, player $i$ choosing a noise distribution mean other than the common one is dominated by choosing the common one.  In particular, if everyone else chooses mean-zero noise, player $i$ should as well.

\end{proof}}

If the parameter $\alpha$ is greater than $\tfrac{1}{2}$, then \textit{every} such equilibrium requires players to choose a mean-zero noise-generating distribution.  To see this, suppose all players $j\neq i$ choose a mean $\mu >  0$.  Then for a small enough value of $\epsilon$, player $i$ choosing a noise-generating distribution with mean $\mu-\epsilon$ is a profitable deviation, since $\alpha>\tfrac{1}{2}$ means that her gain from moving $\noisy_i$ closer to $s$ more than offsets the loss from moving further from $\bar{\noisy}$.  Going forward, we assume without loss of generality that any noisy equilibrium is one in which the added noise comes from a mean-zero distribution.

We now define our equilibrium concept, which is analogous to that of the original game.
\begin{definition}[Symmetric noisy linear Nash equilibrium]
	A  \textit{symmetric noisy linear Nash equilibrium} of this game is a strategy profile $\mathbf{\noisy}$ where each player $i$ chooses
	\begin{align*}
	\noisy_i = \theta_i + \eta_i  
	\end{align*} 
	and $\theta_i = \kappa x_i +(1-\kappa) y$, $\kappa$ is the same for all agents, and each player's noise $\eta_i$ is drawn independently from the same distribution $H$.
	
\end{definition}

One important point is that $\rho(\noisy_i)$ is a function of the \emph{optimal Bayesian posterior} distribution of $x_i$ given $\noisy_i$ and $y$, so an agent must consider his choice of action both in relation to the change in coordination in addition to others' beliefs about him; however, these beliefs are only what an optimal Bayes estimator, knowing the equilibrium profile, could \emph{infer}. In particular, agents are concerned only by the knowledge of agents that are \emph{correct}. It is conceivable that agents could worry about opponents \emph{incorrectly} inferring their signals, but that falls outside the scope of this model.

\subsection{Solving the Extended Game}
\label{sec:solve}
Having defined the privacy-extended game, we can solve for the parameter values which support a symmetric noisy linear Nash equilibrium. Much of the analysis follows either directly from or along similar arguments as in the original game. We proceed as follows. First, we demonstrate in \Cref{lem:sep} that the utility function separates additively into three components: the utility in the original game, a penalty in this utility due to the added noise, and the privacy term.  We then derive the first order condition of a representative agent in \Cref{lem:FOCext} and use this to find the optimal parameter values in \Cref{cor:params}, depending on whether we consider the precision-based or entropy-based measure of obfuscation.  Finally, in \Cref{lem:slnne-var,lem:slnne-ent}, we show that these values support an equilibrium.  This completes the proof of \Cref{thm:main-thm}.

We begin with the separability of the utility function. 
\begin{lemma}[Separability]\label{lem:finite-sep}
	
	In the game with finitely many players, the players' utility functions in the privacy-aware game  separate additively into the utility in the privacy unaware game, a penalty in $\nu_i$, and a privacy term as
	
	\begin{align*}
	v_i(\noisy_i, \noisy_{-i}) &= (1-\beta)\left(-\alpha(\noisy_i -s)^2 -(1-\alpha)(\noisy_i-\bar{\noisy})^2\right) +\beta\rho(\noisy_i) \\
	&= (1-\beta)u_i(\theta_i, \noisy_{-i})  +(1-\beta)\left(\alpha +\left(1-\frac{1}{n}\right)^2 (1-\alpha)\right) \nu_i +\beta \rho(\noisy_i),
	\end{align*}
	where $\nu_i$ denotes the variance of the noise-generating distribution $H_i$ of player $i$ and $u_i$ is the utility function in the privacy-unaware game.
\end{lemma}
\mute{\begin{proof}
	The proof is nearly identical to that of \Cref{lem:sep}. Writing $\noisy_i$ as $\theta_i+\eta_i$, i.e. a deterministic component plus random noise, we can decompose the various pieces of the utility function as follows.  The first part is 
	\begin{align*}
	-\alpha \E_i[(\theta_i +\eta_i-s)^2] = -\alpha\E_i\left[(\theta_i-s)^2+\eta_i(\theta_i-s)+\eta_i^2\right] = - \alpha \E_i[(\theta_i-s)^2] - \alpha \nu_i ,
	\end{align*}
	where we have again used the independence of $\eta_i$ to conclude that $E_i[\eta_i (\theta_i-s)]=0$. 
	
	The second term is 
	\begin{align*}
	-(1-\alpha) \E_i\left[\left(\theta_i -\eta_i -\frac{1}{n} \sum\limits_{j=1}^n \noisy_j\right)^2\right] = -(1-\alpha) \E_i \left[\left(\theta_i\left(1-\frac{1}{n}\right) + \eta_i\left(1-\frac{1}{n}\right) -\frac{1}{n} \sum\limits_{j\neq i}\noisy_j\right)^2\right]
	\end{align*}
	Rewriting gives: 
	\begin{align*}-(1-\alpha)\E_i \left[\left(\theta_i\left(1-\frac{1}{n}\right) - \frac{1}{n} \sum_{j\neq i} \noisy_j\right)^2 \right] -(1-\alpha) \left(1-\frac{1}{n}\right) \nu_i -(1-\alpha)\E_i\left[\eta_i\left(1-\frac{1}{n}\right)^2\frac{1}{n}\theta_i\sum_{j\neq i} \noisy_j\right]
	\end{align*}
	After collecting terms, we see that it suffices to argue that the final term
	\begin{align*}
	(1-\alpha)\E_i\left[\eta_i\left(1-\frac{1}{n}\right)^2\frac{1}{n}\theta_i\sum_{j\neq i} \noisy_j\right] =0,
	\end{align*}
	which follows from the fact that $\eta_i$ is independent of all the other parameters of the game as well as the other $\eta_j$.
	
\end{proof}}

\begin{corollary}
	\label{lem:sep}
	Suppose that all agents play a noisy strategy $\noisy_i = \theta_i+ \eta_i$, with $\eta_i$ being a random variable and $\mathbb{E}(\eta_i)=0$. Then an agent's utility can be decomposed into
	\begin{align*}
	\E_i [v_i(\noisy,\noisy_{-i})] =  (1-\beta) \E_i [u(\theta_i, \noisy_{-i})] + (1-\beta) \nu_i + \beta \rho(\noisy_i)
	\end{align*}
	where $\nu_i$ denotes the variance of the noise-generating distribution $H_i$ of player $i$ and $u_i$ is the utility function in the privacy-unaware game.
\end{corollary}

This shows that agents can evaluate what their optimal action would be in the original game and then decide the optimal amount of noise to add, choosing their action as the realized value of the noise plus their optimal action.

We next derive the first order conditions for the privacy-extended game. { By  \Cref{lem:finite-sep} and Corollary \ref{lem:sep}, the first order conditions can be disentangled into a first order condition on the deterministic component, which must be as in \Cref{lem:FOC} and \Cref{lem:finite-foc}, and a separate first order condition on the variance of the random component. Because the added noise is mean-zero, expectations about average action and state are as before, and the $\kappa$ in the optimal deterministic action is the same as before. } 

\begin{lemma}[Privacy-aware first order conditions]\label{lem:finite-ext-foc}
	
	In an equilibrium of the game with finitely many players where the optimal action is $\noisy^*_i = \theta^*_i + \eta_i$,  the optimal choice of $\theta^*_i$ and the variance $\nu^*$ for the noise-generating distribution $H_i$ from which $\eta_i$ is drawn must satisfy
	
	\begin{align*}
	\theta_i^* &= \frac{\alpha n^2 \E_i[s]}{\alpha (2n-1)+(n-1)^2} + \frac{(1-\alpha)(n-1)\E_i\left[\sum\limits_{j\neq i}\theta_{j}\right]}{\alpha (2n-1)+(n-1)^2}, \\ 	
	\frac{\partial \rho}{\partial \nu^*} &=-\frac{-(1-\beta)\left(\alpha +\left(1-\frac{1}{n}\right)^2 (1-\alpha)\right)}{\beta}.
	\end{align*}

\end{lemma}

\mute{\begin{proof}
	Using \Cref{lem:finite-sep}, we can decompose the expected utility of agent $i$ as
	\begin{align*}
	\E_i [v_i(\noisy_i, \noisy_{-i} ] = (1-\beta)\E_i[u_i(\theta_i, \noisy_{-i})]  -(1-\beta)\left(\alpha +\left(1-\frac{1}{n}\right)^2 (1-\alpha)\right) \nu_i +\beta \rho(\noisy_i),
	\end{align*}
	
	which is the sum of a piece that depends on $\theta_i$ and a piece that depends on $\nu_i$.
	
	The agent can therefore optimize each piece separately with her choice of $\theta_i$ and $\nu_i$. \Cref{lem:finite-foc} gives the first order condition on $\theta_i^*$.
	
	To find the first order condition on $\nu^*$, we can write
	
	\begin{align*}
	0=\frac{\partial v_i}{\partial \nu^*} = -(1-\beta)\left(\alpha +\left(1-\frac{1}{n}\right)^2 (1-\alpha)\right) +\beta\frac{\partial\rho}{\partial \nu^*}
	\end{align*}
	
	and solve for $\frac{\partial\rho}{\partial\nu^*}$ to get the result.

\end{proof}	}

\begin{proposition}\label{lem:FOCext}
In the game with infinitely many players, in an equilibrium where the optimal action is $\noisy^*_i = \theta^*_i + \eta_i$,  the optimal choice of $\theta^*_i$ and the variance $\nu^*$ for the noise-generating distribution $H_i$ from which $\eta_i$ is drawn must satisfy

	\begin{align*}
	\theta_i^* = \alpha \E_i[s] +(1-\alpha) \E_i[\bar{\noisy}] \qquad 	\frac{\partial \rho}{\partial \nu^*} = -\frac{1-\beta}{\beta}.
	\end{align*}

	\end{proposition}
\mute{\begin{proof}
	Using Lemma \ref{lem:sep}, we can decompose the utility of agent $i$ as
	\begin{align*}
	\E_i [v_i(\noisy,\noisy_{-i})] =  (1-\beta) \E_i [u(\theta_i, \noisy_{-i})] - (1-\beta) \nu_i - \beta \rho(\noisy_i),
	\end{align*}
	
	which is the sum of a piece that depends on $\theta_i$ and a piece that depends on $\nu_i$.
	
	The agent can therefore optimize each piece separately with her choice of $\theta_i$ and $\nu_i$. \Cref{lem:FOC} gives the first order condition on $\theta_i^*$.
	
	To find the first order condition on $\nu^*$, we can write
	
	\begin{align*}
	0=\frac{\partial v_i}{\partial \nu^*} = -(1-\beta) -\beta\frac{\partial\rho}{\partial \nu^*}
	\end{align*}
	
	and solve for $\frac{\partial\rho}{\partial\nu^*}$ to get the result.

\end{proof}}

\begin{corollary}[Finite game privacy parameters]
	\label{cor:finiteparams}
	The optimal deterministic component $\theta_i$ is, as before, 
	\begin{align*}
	\theta_i^* =  \kappa x_i + (1-\kappa) y
	\end{align*}
	and the optimal choice of variance for the noise distribution is
	\begin{align*}
	\nu_{i,prec}^* = \sqrt{\frac{\beta}{1-\beta} \left(\alpha + (1-\alpha)\left(1-\frac{1}{n}\right)^2\right)} \ \  \text{ or } \ \  \nu_{i,ent}^* = {\frac{\beta}{1-\beta} \left(\alpha + (1-\alpha)\left(1-\frac{1}{n}\right)^2\right)} 
	\end{align*}
	where $\nu_{i,prec}^*$ and $\nu_{i,ent}^*$ are the optimal variances under $\rho$ being the precision and entropy privacy measures, respectively.  
\end{corollary}

\begin{corollary}[Infinite game privacy parameters]
	\label{cor:params}
	The optimal deterministic $\theta_i$ is, as before, 
	\begin{align*}
	\theta_i^* =  \kappa x_i + (1-\kappa) y
	\end{align*}
	and the optimal choice of variance for the noise distribution is
	\begin{align*}
	\nu_{i,prec}^* = \sqrt{\frac{\beta}{1-\beta}} \qquad \nu_{i,ent}^* = \frac{\beta}{1-\beta}
	\end{align*}
	where $\nu_{i,prec}^*$ and $\nu_{i,ent}^*$ are the optimal choices of variance when $\rho$ is the precision-based or entropy-based privacy measure, respectively.  
\end{corollary}

Before proving this, we state a fact about the entropy of Gaussian distributions.

\begin{fact}
	\label{fact:maxent}
	Among all distributions supported on the entire real line with a fixed mean $\mu$ and variance $\sigma^2$, the Gaussian $\mathcal{N}(\mu,\sigma^2)$ achieves the maximum entropy, and its entropy is given by
	$\tfrac{1}{2}\log{(2\pi e \sigma^2)}.$
\end{fact}

Since agents pay a penalty in the variance of their noise distribution, agents who measure their privacy loss with entropy pay the same penalty for any distribution with a fixed variance, and their gain from privacy is maximized by picking the distribution with maximum entropy.  Therefore, the choice of a mean-zero Gaussian with appropriate variance is the dominant strategy for such agents.
A proof of this fact can be found in Chapter 9 of \cite{cover2012elements}.

\begin{proof}[Proof of \Cref{cor:finiteparams,cor:params}]
	The fact that the first order condition on $\theta_i^*$ is the same as that of $\theta_i^*$ in the original game, and that agents add mean-zero noise (implying  $\E_i[\bar{\noisy}]=\E_i[\bar{\theta}]$) implies the optimal choice of $\theta_i^*$. To find the optimal $\nu^*$, we note that the reciprocal variance and  entropy penalties (respectively) are
	\begin{align*}
	\rho_{prec}(\nu) = \frac{-1}{\nu} \qquad \rho_{ent}(\nu) = \frac{1}{2}  \log{( 2 e\pi \nu)}.
	\end{align*}
	
Differentiating these with respect to $\nu$ and rearranging gives the results.
\end{proof}

	We can observe that the optimal choice of $\nu_i^*$ is similar for both of these functions, even if their constructions are not.  Each is increasing at a decreasing rate as $\nu$ grows, since their derivatives are of the form $\frac{1}{\nu^c}$ for $c\geq 1$.  Since players gain a diminishing marginal benefit for adding additional noise as they increase $\nu$ but pay a penalty linear in $\nu$, we should expect this trade-off to point to an equilibrium.  This will hold for any penalty function whose derivative is of this type, although the interpretation of such a function may not be as natural of a property of a distribution as precision or entropy.

Using these, a symmetric linear noisy Nash equilibrium of this game follows:

\begin{lemma}[Equilibrium -- precision loss]
	\label{lem:slnne-var}
		There exists a symmetric linear noisy Nash equilibrium in which all agents use $\kappa$ as defined in \Cref{lem:finitekappa} and draw independent noise from a distribution $H_i$ with variance $\nu^*$,
	where $H_i$ can be any distribution with mean zero and variance $\nu_{i,prec}^* = \sqrt{\frac{\beta}{1-\beta}}$. Here, agents may choose \emph{any} distribution with these properties.
\end{lemma} 

\begin{lemma}[Equilibrium -- entropy loss]
	\label{lem:slnne-ent}
	There exists a symmetric linear noisy Nash equilibrium in which all agents use $\kappa$ as defined in \Cref{lem:finitekappa} and draw independent noise from the Gaussian distribution $N(0, \nu_{ent}^*)$.
	
	Since the Gaussian is the maximum entropy distribution with fixed variance $\nu^*$, all agents draw noise (independently) from the same distribution.
\end{lemma}

This completes the proof of \Cref{thm:main-thm}, since we have found a symmetric linear noisy action for each player which satisfies the first order conditions of the privacy aware game.  We next analyze and discuss the \textit{price of privacy} in the new game.

\section{Price of Privacy}
\label{sec:pop}
  Informally, the price of privacy describes the loss in quality of some measure as a result of introducing privacy-awareness into the game.  We describe two different, but related, quantities which represent this effect and discuss some settings where one may be interested in each.

\subsection{The Agents' Cost}

Our first notion of the price of privacy can be viewed as the utility the \textit{agents} pay in the original game in order to express their value for obfuscation.  Formally, recall that $u_i(\theta_i,\theta_{-i}) = -(1-\alpha)(\theta_i-\bar{\theta})^2 - \alpha(\theta_i-s)^2$ is the utility function in the unmodified game. In the symmetric linear Nash equilibrium, each player chooses the prescribed optimal action $\theta_i^*$ and earns utility $u_i(\theta_i^*,\theta_{-i}^*)$.  Similarly, in the modified game, in the symmetric linear Nash equilibrium chooses the optimal noisy action $\noisy_i^*$.  We can now ask how much worse-off playing the actions $\noisy^*$ in the original game makes a representative player as compared to playing the actions $\theta^*$.  Formally, we can quantify the price of privacy by taking the ratio of these utilities.

\begin{definition}[Agents' price of privacy]
	The \emph{price of privacy} given an information structure is 
	\begin{align*}
	\PoP(\tau_x,\tau_y,\beta) = \frac{\E_i [u_i(\noisy_i^*,\noisy_{-i}^*)]}{\E_i [u_i(\theta_i^*,\theta_{-i}^*)]}
	\end{align*}
where the expected utility is with respect to the game with signal variances $\sigma^2_x$ and $\sigma^2_y$.
\end{definition}

\begin{lemma}[Agents' price of privacy -- form]\label{lem:popform}
	The price of privacy in the game where agents play a linear strategy has the form
\begin{align*} 
\PoP(\tau_x,\tau_y,\beta) =1 +\frac{\nu^*_i}{ \E_i [u_i(\theta_i^*,\theta_{-i}^*)]}
\end{align*}
where the expected utility is with respect to the game with signal variances $\sigma^2_x$ and $\sigma^2_y$.
\end{lemma}
\mute{\begin{proof}
	First note that 
\begin{align*}
\E_i[u_i(\theta_i, \noisy_{-i})]= \E_i[u_i(\theta_i, \theta_{-i})] 
\end{align*}
because the noise added to $\theta_i$ has a mean of zero.  Now, \Cref{lem:sep} lets us write
\begin{align*}
\PoP(\tau_x,\tau_y,\beta) = \frac{\E_i [u_i(\noisy_i^*,\noisy_{-i}^*)]}{\E_i [u_i(\theta_i^*,\theta_{-i}^*)]}  = \frac{\E_i[ u_i(\noisy_i^*,\noisy_{-i}^*)] + \nu^*_i }{\E_i[ u_i(\theta_i^*,\theta_{-i}^*)]}
\end{align*}
where we have factored out all of the negative signs.  Combining with the previous part, we have that 
\begin{align*}
\PoP(\tau_x,\tau_y,\beta) = \frac{\E_i[ u_i(\noisy_i^*,\noisy_{-i}^*)] + \nu^*_i }{\E_i[ u_i(\theta_i^*,\theta_{-i}^*)]} = \frac{\E_i[ u_i(\noisy_i^*,\theta{-i}^*)] + \nu^*_i }{\E_i[ u_i(\theta_i^*,\theta_{-i}^*)]} = 1 + \frac{\nu^*_i}{\E_i [u_i(\theta_i^*,\theta_{-i}^*)]},
\end{align*}
as desired.
\end{proof}}

\begin{theorem}[Price of privacy -- value]\label{thm:pop}

	In the game with finitely many agents, the price of privacy is:
	\begin{align*}
	\PoP(\tau_x,\tau_y,\beta) = 1+\frac{\nu^*_i}{\alpha (\kappa^2 \sigma^2_x + (1-\kappa)^2\sigma^2_y) +\frac{(1-\alpha)\kappa^2(n-1)}{n}\left.\sigma^2_x \right.},
	\end{align*}
	
	for 	
	\begin{align*}
	\nu_{i,prec}^* = \sqrt{\frac{\beta}{1-\beta} \left(\alpha + (1-\alpha)\left(1-\frac{1}{n}\right)^2\right)} \ \  \text{ or }\ \  \nu_{i,ent}^* = {\frac{\beta}{1-\beta} \left(\alpha + (1-\alpha)\left(1-\frac{1}{n}\right)^2\right)} \end{align*} 
	if we measure the privacy loss using precision or entropy, respectively, and 
	
	\begin{align*}
	\kappa = \frac{\alpha n^2 \tau_x}{\alpha n^2 \tau_x + \left((n-1)^2 +\alpha (2n-1)\right)\tau_y}.
	\end{align*}

\end{theorem}

\begin{proposition}
	
	In the game with infinitely many agents, the price of privacy is:
	\begin{align*}
	\PoP(\tau_x,\tau_y,\beta) = 1+\frac{\nu^*_i}{(\alpha(1-\kappa)^2\sigma^2_y + \kappa^2\sigma^2_x)},
	\end{align*}
	
	where $\nu^*_i$ is $\sqrt{\frac{\beta}{1-\beta}}$ if privacy is measured by precision and $\frac{\beta}{1-\beta}$ if privacy is measured by entropy and recalling that $\kappa= \frac{\tau_x \alpha}{\tau_x \alpha +\tau_y}$.
	
\end{proposition}

\begin{proof}
	This follows directly from the expected utility computation in \Cref{lem:eukbp}  and the functional form of the price of privacy in \Cref{lem:popform}.
\end{proof}

We can observe several features of the price of privacy. First, it can be \emph{arbitrarily large}, depending on the value of the parameters.  If $\beta$ is close to 1, then agents have a relatively high value for obfuscation and will add large amounts of noise to their actions in order to protect their private signals.  As $\beta$ becomes closer to zero, agents do not care too much about privacy and their actions become more and more similar to those in the original game, so the price of privacy is minimized when $\beta=0$.  Second, if we fix a value of $\beta$, then the price of privacy decreases as the variance of the public and private signals \textit{increase}.  This is because the value of $\beta$ determines the noise that players will add to their signals, and therefore fixes the numerator of the fractional part of the price of privacy.  The price of privacy is \textit{decreasing} in any factor that improves the expected utility, such as decreasing the signal variances or $n$ (as in \Cref{cor:compstat}).  The effect is ambiguous when changing $\sigma^2_y$, due to the over-weighting phenomenon.

This ratio measures the degree to which agents are worse-off in the coordination and guessing components by playing the privacy-aware equilibrium noisy actions as compared to the privacy-unaware equilibrium actions for a fixed $\alpha$, $\sigma^2_x$, and $\sigma^2_y$ and the realizations of $x_i$ and $y$.  This can be thought of as describing the increased risk of making the `wrong' decision in an opinion-aggregation setting or of misvaluing an asset in a financial markets setting as a result of agents adding noise to their actions.

\subsection{The Aggregator's Cost}

Suppose instead we are viewing this game from the position of a data aggregator.  The aggregator can observe the actions of some finite number of agents $n$ and takes the average of these to estimate the true state of the world $s$ by taking a simple average.  The aggregator does not know the realizations of any of the $x_i$ or $y$, but she does know the signal variances and the value of $\kappa$;  she also knows that $\E[\theta_i]=\E[\noisy_i]=s$ for all agents $i$, so the sample average will provide an unbiased estimate of $s$.
If we define the aggregator's `utility' to be the variance of its sample about $s$, then we can make the following observation:

\begin{lemma}[Aggregator's utility]
	\label{lem:util-agg}
	Consider an instance of the privacy-aware game where an aggregator observes the actions of $n$ agents (either all all of the agents in the finite case or some uniformly random sample in the finite or infinite case), the  signal variances are $\sigma^2_x$ and $\sigma^2_y$, and players choose to add mean-zero noise with variance $\nu^*_i$.  Then the utility of the aggregator, as measured by the variance of the sample average about the true state $s$ is given by
	\begin{align*}
	\mathcal{U}_{agg}(\sigma^2_x,\sigma^2,\nu^*_i,n) &=\E\left[  \left(\left.\left(\frac{1}{n}\sum\limits_{i=1}^n \noisy_i\right)  - s \right)^2       \right\vert s \right]\\
	&=\frac{\kappa^2}{n}\sigma^2_x + \frac{\nu^*_i}{n} + (1-\kappa)^2\sigma^2_y.
	\end{align*}
	
	We can find the aggregator's utility in the privacy-unaware game by letting $\nu^*_i=0$, which recovers a result in \citet{morrisshin}.
\end{lemma}

\mute{\begin{proof}

\begin{align*}
\mathcal{U}_{agg}(\sigma^2_x,\nu_i^*,\sigma^2,n) &=\E\left[  \left(\left.\left(\frac{1}{n}\sum\limits_{i=1}^n \noisy_i\right)  - s \right)^2       \right\vert s \right]\\
&=\E\left[  \left(\left.\left(\frac{1}{n}\sum\limits_{i=1}^n \kappa x_i + (1-\kappa)y + \eta_i \right)  - s \right)^2       \right\vert s \right]\\
&=\E\left[  \left(\left.\left(\frac{1}{n}\sum\limits_{i=1}^n \kappa x_i + \eta_i \right) + (1-\kappa)y  - s \right)^2       \right\vert s \right]\\
&=\E\left[  \left(\left.\left(\frac{1}{n}\sum\limits_{i=1}^n \kappa (s + \varepsilon_{x_i} ) +\eta_i \right) + (1-\kappa)(s+\varepsilon_y)  - s \right)^2       \right\vert s \right]\\
&=\E\left[  \left(\left.\left(\frac{1}{n}\sum\limits_{i=1}^n \kappa\varepsilon_{x_i} + \eta_i \right) + (1-\kappa)\varepsilon_y  \right)^2       \right\vert s \right]\\
&=\E\left[  \left(\frac{\kappa}{n}\sum\limits_{i=1}^n \varepsilon_{x_i} + \frac{1}{n} \sum\limits_{i=1}^{n} \eta_i  + (1-\kappa)\varepsilon_)   \right)^2        \right]\\
&=\frac{\kappa^2}{n}\sigma^2_x + \frac{\nu^*_i}{n} + (1-\kappa)^2\sigma^2_y
\end{align*}

where we have decomposed $x_i$ and $y$ into $s$ plus mean-zero Gaussian noise $\varepsilon_{x_i}$ and $\varepsilon_y$.
\end{proof}}

We can again take the ratio of the aggregator's utility in the privacy-aware game to that in the privacy-unaware game to quantify the extent to which privacy-awareness degrades the quality of the aggregator's sample mean as an estimate of $s$.

\begin{lemma}[The Aggregator's price of privacy]
	
	The price of privacy for the aggregator in equilibrium is given by
	
	\begin{align*}
	\PoP_{agg}(\sigma^2_x,\sigma^2_y,n) = \frac{\E[\mathcal{U}_{agg}(\sigma^2_x,\sigma^2_y,\nu^*_i,n)]}{\E[\mathcal{U}_{agg}(\sigma^2_x,\sigma^2_y,0,n)]}
	=1+\frac{\nu^*_i}{\kappa^2\sigma^2_x + n(1-\kappa)^2\sigma^2_y}
	\end{align*}
		where $\nu^*_i$ is $\sqrt{\frac{\beta}{1-\beta}}$ if privacy is measured by precision and $\frac{\beta}{1-\beta}$ if privacy is measured by entropy and recalling that $\kappa= \frac{\tau_x \alpha}{\tau_x \alpha +\tau_y}$ in the game with infinitely many agents and $\kappa = \frac{\alpha n^2 \tau_x}{\alpha n^2 \tau_x + \left((n-1)^2 +\alpha (2n-1)\right)\tau_y}$ in the game with finitely many agents.
\end{lemma}
\begin{proof}
	This follows from plugging in the forms for the expected utilities from \Cref{lem:util-agg} and simplifying.
\end{proof}

We can again observe that, for the same reason as the agents' price of privacy, that this quantity can be arbitrarily large as $\beta$ approaches one and that it is decreasing as the signal variances grow.  We can also observe that this price decreases towards one as $n$ grows, all else fixed. This is consistent with the intuition that an aggregator can offset the cost of agents adding more noise to their actions by making more observations.  

The agents' and aggregator's prices of privacy are somewhat similar both in functional form and in interpretation and this results from their values almost aligning.  The agents' cost is measured by a weighted deviation from both the true state and the average while the aggregator's is measured by a deviation from the true state alone.  Any configuration of the game parameters which cause $\kappa=1$ (i.e. agents ignore the public signal) causes them to be exactly equal.  This again highlights the problem of over-weighting the public signal, and this issue is exacerbated for the aggregator, who, in effect, observes $n$ `copies' of $y$ rolled into the actions $\theta_i$ and $\noisy_i$.  If the aggregator knows the realization of $y$, she can, similarly to the agents, use that information to find (an estimate of) the private signals $x_i$ from the observations of $\theta_i$ or $\noisy_i$.  This motivates the same privacy concern as in the case where other agents were trying to infer agent $i$'s private signal, this time in the context of a distrusted central aggregator.

\section{Discussion and Future Directions}
In this work, we used a game-theoretic model to examine a game with agents acting to 
obscure their private information.  We began with the Keynesian Beauty Contest and modified the agents' utility functions so as to endogenize a notion of privacy.  Using this, we can quantify the social costs of this `selfish' privacy protection as a `price of privacy', both with respect to the agents as well as an (untrusted) central aggregator.

A clear next-step would be to repeat this analysis for other stylized models of strategic interaction, such as bargaining games or resource allocation problems, where individuals seek to balance expressing their preferences well enough to achieve a high payoff but choose to deviate slightly from this in order to obscure their true valuations.

Finally, the high-level motivation in this work is to examine how strategic agents behave when they are concerned with privacy, so another interesting direction would be to integrate this approach with existing formal notions of privacy.  For example, examining what agents' utility functions and action spaces must look like in order for there to exist an equilibrium which implements a locally differentially private mechanism could be a fruitful avenue of research.

\ifarxiv

\subsection*{Acknowledgments}

We would like to graciously thank Annie Liang and Aaron Roth for their guidance in the research process and feedback on previous drafts of this project.  We would additionally like to thank the anonymous reviewers for their helpful feedback.
\bibliographystyle{plainnat}
\fi
\bibliography{bib}

\ifarxiv
\pagebreak
\appendix
\section{Omitted Proofs}\label{apx:omitted}

We present here the statements and proofs of the results whose proofs were omitted from the main text.  We also present from first-principles the proofs of the analogous results in the game with infinitely many players, which demonstrates that the proofs where we took the limit of the finite result as the number of players grows to infinity are indeed correct.

\subsection*{Omitted Proofs from \Cref{sec:framework}}
\begin{lemma*}[\ref{lem:finite-foc}]
	In equilibrium, an agent's optimal choice of $\theta_i^*$ must satisfy
\begin{align*}
\theta_i^* = \frac{\alpha n^2 \E_i[s]}{\alpha (2n-1)+(n-1)^2} + \frac{(1-\alpha)(n-1)\E_i\left[\sum\limits_{j\neq i}\theta_{j}\right]}{\alpha (2n-1)+(n-1)^2}
\end{align*}

\end{lemma*}

\begin{proof}
	Fix the actions of all other players $j\neq i$ and consider the utility of player $i$, recalling that in the finite setting, 
	$$\bar{\theta}= \frac{1}{n}\theta_i^* + \frac{1}{n}\sum\limits_{j\neq i}\theta_{j}.$$
	
	We have
	\begin{align*}
	u_i(\theta_i^*,\theta_{-i}) 
	&= -\alpha(\theta_i^*-s)^2 - (1-\alpha)(\theta_i^*-\bar{\theta})^2\\
	&= -\alpha(\theta_i^*-s)^2 - (1-\alpha)\left(\theta_i^*-\left(\frac{1}{n}\theta_i^* + \frac{1}{n}\sum\limits_{j\neq i}\theta_{j}\right)\right)^2\\
	&= -\alpha(\theta_i^*-s)^2 - (1-\alpha)\left(\frac{n-1}{n}\theta_i^*-\left( \frac{1}{n}\sum\limits_{j\neq i}\theta_{j}\right)\right)^2\\
	&= -\alpha(\theta_i^*-s)^2 - (1-\alpha)\frac{1}{n^2}\left((n-1)\theta_i^*-\left(\sum\limits_{j\neq i}\theta_{j}\right)\right)^2
	\end{align*}

	Player $i$'s first order condition will be to maximize this in expectation. Taking an expectation, we have 
	\begin{align*}
	\E_i[u_i(\theta_i,\theta_{-i})] 
	&= \E_i\left[ -\alpha(\theta_i^*-s)^2 - (1-\alpha)\frac{1}{n^2}\left((n-1)\theta_i^*-\left(\sum\limits_{j\neq i}\theta_{j}\right)\right)^2 \right]\\
	&= -\alpha\E_i[(\theta_i^*-s)^2] - (1-\alpha)\frac{1}{n^2} \E_i\left[\left((n-1)\theta_i^*-\left(\sum\limits_{j\neq i}\theta_{j}\right)\right)^2\right]\\
	&= -\alpha\E_i[(\theta_i^*)^2-2s\theta_i^*+s^2]\\
	&\qquad - (1-\alpha)\frac{1}{n^2} \E_i\left[ (n-1)^2(\theta_i^*)^2 -2(n-1)\theta_i^* \left(\sum\limits_{j\neq i}\theta_{j}\right) + \left(\sum\limits_{j\neq i}\theta_{j}\right)^2 \right]
	\end{align*}

	Differentiating with respect to agent $i$'s choice of $\theta_i$ gives us that, in equilibrium, $\theta_i^*$ must satisfy
	\begin{align*}
	\frac{\partial}{\partial \theta_i^*} &= 0 =
	-\alpha(2\theta_i^* - 2\E_i[s]) - (1-\alpha)\frac{1}{n^2} \left(2(n-1)^2\theta_i^* - 2(n-1)\E_i\left[\sum\limits_{j\neq i}\theta_{j}\right]\right)\\
	&= \alpha(\theta_i^* - \E_i[s]) + \frac{1-\alpha}{n^2} \left((n-1)^2\theta_i^* - (n-1)\E_i\left[\sum\limits_{j\neq i}\theta_{j}\right]\right)\\
	&= \theta_i^*\left( \alpha + \frac{(1-\alpha)(n-1)^2}{n^2}  \right) - \left(\alpha\E_i[s] + \frac{(1-\alpha)(n-1)}{n^2}\E_i\left[\sum\limits_{j\neq i}\theta_{j}\right]\right)
	\end{align*}
	
	Where we have used the fact that $\E_i[\theta_i^*]=\theta_i^*$, since it is agent $i$'s choice and not a random variable.
	
	Rearranging yields the claim.

\end{proof}

\begin{proposition*}[\ref{lem:FOC}]
	In the game with infinitely many agents, in equilibrium, an agent's optimal choice of $\theta_i^*$ must satisfy
	\begin{align*}
	\theta_i^* = \alpha \E_i[s]+(1-\alpha) \E_i[\bar{\theta}]
	\end{align*}
\end{proposition*}\begin{proof}
	Since agent $i$ chooses his action after observing signals $x_i$ and $y$, they optimize their action given those signals. We can write:
	\begin{align*}
	\theta_i^* \in \argmax \E_i [u_i(\theta_i)]
	\end{align*}
	Expanding, we have:
	\begin{align*}
	\E_i [u_i(\theta_i,\theta_{-i})] =-\alpha(\E_i[\theta_i^2 -2 \theta_i s+ s^2]) -(1-\alpha)(\E_i[\theta_i^2 -2\theta_i \bar{\theta} +\bar{\theta}^2]).
	\end{align*}
	We can differentiate with respect to $\theta_i^*$ and set the derivative equal to zero to find
	\begin{align*}
	0=\frac{\partial}{\partial \theta_i} \E_i[u_i(\theta_i,\theta_{-i})] = -\alpha (2\theta_i - 2 \E_i[s])-(1-\alpha)[2\theta_i -\bar{\theta}].
	\end{align*}
	Where we have used the fact that $\E_i[\theta_i^*]=\theta_i^*$, since it is agent $i$'s choice and not a random variable, as well as the fact that $\frac{\partial\bar{\theta}}{\partial\theta_i^*}=0$ because there are infinitely many agents, so agent $i$'s action cannot affect the average of all of the actions.

	Rearranging yields the claim. 
\end{proof}
{
\begin{remark*}
	These results show that the optimal deterministic action \emph{is} a linear combination of expectations about the state and average action; thus, the optimal deterministic action and the optimal linear action are one and the same. Because these propositions hold also for the first-order conditions in the extended game, this will be true in the extended game as well. 
\end{remark*}
}
\begin{remark*}
	If we think of $\sum\theta_j$ as $(n-1)$ times the average action of other players and take the limit of this expression as $n$ goes to infinity, we recover the result for the continuum case, as the coefficient on $\E_i[s]$ becomes $\alpha$ and the coefficient on the expected average action becomes $(1-\alpha)$.
\end{remark*}

\begin{lemma*}[\ref{lem:finitekappa}]
	In the game with finitely many agents, the value of $\kappa$ which supports a symmetric linear Nash equilibrium is
 \begin{align*}
\kappa = \frac{\alpha n^2 \tau_x}{\alpha n^2 \tau_x + \left((n-1)^2 +\alpha (2n-1)\right)\tau_y}
 \end{align*}
\end{lemma*}

\begin{proof}
	
	Suppose that all agents play $\theta_i = \kappa_i + (1-\kappa)y$, and consider a representative agent $i$.
	
	Since agent $i$ is Bayesian, she aggregates the public and private signals according to their precisions to compute $$\E_i[s] = \frac{\tau_x x_i + \tau_y y}{\tau_x+\tau_y}$$ and her belief about any other agent $j$'s private signal $x_j$ is that $$\E_i[x_j] = \E_i[s].$$

Supposing that each agent acts optimally and chooses $\theta^*_i$ according to \Cref{lem:finite-foc}, agent $i$ will write the following.  We let $\mathcal{C} = \frac{1}{\alpha (2n-1)+(n-1)^2}$ for clarity of notation.

\begin{align*}
\theta_i^*  &= \frac{\alpha n^2 \E_i[s]}{\alpha (2n-1)+(n-1)^2} + \frac{(1-\alpha)(n-1)\E_i\left[\sum\limits_{j\neq i}\theta_{j}\right]}{\alpha (2n-1)+(n-1)^2}\\
&= \mathcal{C} \left(\alpha n^2\E_i[s] + (1-\alpha)(n-1)\E_i\left[\sum\limits_{j\neq i}\left(\kappa x_j + (1-\kappa)y\right)\right]\right)\\
&= \mathcal{C} \left(\alpha n^2\E_i[s] + (1-\alpha)(n-1)\left[\sum\limits_{j\neq i}\left(\kappa \E_i[s] + (1-\kappa)y\right)\right]\right)\\
&= \mathcal{C} \left(\alpha n^2\E_i[s] + (1-\alpha)(n-1)^2\left(\kappa \E_i[s] + (1-\kappa)y\right)\right)\\
&= \mathcal{C} \left(\alpha n^2\frac{\tau_x x_i + \tau_y y}{\tau_x+\tau_y} + (1-\alpha)(n-1)^2\left(\kappa \frac{\tau_x x_i + \tau_y y}{\tau_x+\tau_y} + (1-\kappa)y\right)\right)\\
&= \mathcal{C}\left(\alpha n^2 \frac{\tau_x}{\tau_x+\tau_y} +(1-\alpha)(n-1)^2\kappa\frac{\tau_x}{\tau_x+\tau_y} \right)  x_i +\\&\qquad + \mathcal{C}\left(   \alpha n^2 \frac{\tau_y}{\tau_x+\tau_y}+(1-\alpha)(n-1)^2 \left(  \kappa \frac{\tau_y}{\tau_x+\tau_y} + (1-\kappa) \right) \right) y
\end{align*}

We can then equate $\kappa$ and the coefficient on $x_i$ or the coefficient on $y$ to $(1-\kappa)$.  By construction, it's easy to see that the sum of these two coefficients will be one, so it suffices to perform the computation for $\kappa$.

\begin{align*}
\kappa &= \mathcal{C}\left(\alpha n^2 \frac{\tau_x}{\tau_x+\tau_y} +(1-\alpha)(n-1)^2\kappa\frac{\tau_x}{\tau_x+\tau_y} \right)\\
&= \mathcal{C}\alpha n^2 \frac{\tau_x}{\tau_x+\tau_y} + \kappa\mathcal{C}(1-\alpha)(n-1)^2 \frac{\tau_x}{\tau_x+\tau_y}\\
&= \frac{\mathcal{C}\alpha n^2 \frac{\tau_x}{\tau_x+\tau_y} }{1-\mathcal{C} (1-\alpha)(n-1)^2 \frac{\tau_x}{\tau_x+\tau_y}}
\end{align*}

Plugging the value for $\mathcal{C}$ back in and simplifying, we get
\begin{align*}
\kappa = \frac{\alpha n^2 \tau_x}{\alpha n^2 \tau_x + \left((n-1)^2 +\alpha (2n-1)\right)\tau_y}
\end{align*}

\end{proof}
\begin{proposition*}[\ref{cor:infintenash}]
	The symmetric linear Nash equilibrium in the game with infinitely many players is given by
	
	\begin{align*}
	\theta_i &= \kappa x_i + (1-\kappa)y \text{, where} \\
	\kappa&= \frac{\tau_x \alpha}{\tau_x \alpha +\tau_y},\ \tau_x  = \frac{1}{\sigma_x^2},\text{ and } \tau_y =\frac{1}{\sigma_y^2}\\
	&\text{for all players } i.
	\end{align*}
	
\end{proposition*}

\begin{proof}
\Cref{lem:FOC} shows that agent $i$'s optimal action is a convex combination of his belief about the true state $s$ and his belief about the average action of all players $\bar{\theta}$. $\E_i[s]$ is independent of the equilibrium profile, and given by
	\begin{align*}
	\E_i[s] = \frac{\tau_x x_i +\tau_y y}{\tau_x +\tau_y},
	\end{align*}
	which is the standard prior-free aggregation of independent Gaussian signals.

	On the other hand, $\E_i[\bar{\theta}]$ does depend on the equilibrium profile. In a SLNE, all other players $j\neq i$ play \begin{align*}
	\theta_j = \kappa x_j +(1-\kappa) y.
	\end{align*}
	Since agent $i$ is Bayesian:
	\begin{align*}
	\E_i[x_j] = \E_i[s] 
	\end{align*}
	and
	\begin{align*}
	\E_i[\theta_j] = \E_i[\kappa x_j + (1-\kappa)y] = \kappa \E_i[x_j] + (1-\kappa)y.
	\end{align*}
	Notice that  the expectation of agent $i$ about the belief of any representative player $j\neq i$ is the same, since his information is symmetric. Thus, we write $\E_i[\theta_{-i}]$ to emphasize that this is the belief of player $i$ about any other player.
	Thus:
	\begin{align*}
	\E_i[\bar{\theta}] = \E_i \int_0^1 \theta_j dj =  \int_0^1 \E_i [\theta_{-i}]  =\E_i[\theta_{-i}] =\kappa \E_i[x_{-i}] + (1-\kappa)y
	\end{align*}
	where we have used the Fubini-Tonelli theorem to exchange the integral with the expectation.
	
	But now note that 
	\begin{align*}
	\theta_i^* &= \alpha \E_i[s] +(1-\alpha) [\kappa \E_i[s]+(1-\kappa) y] \\&=\alpha \frac{\tau_x x_i +\tau_y y}{\tau_x +\tau_y} +(1-\alpha) (\kappa \frac{\tau_x x_i +\tau_y y}{\tau_x +\tau_y}) +(1-\alpha)(1-\kappa) y
	\end{align*}
	We can rearrange this as:
	\begin{align*}
	\theta_i^* &= \left[\frac{\alpha \tau_x}{\tau_x+\tau_y} +(1-\alpha) \frac{\kappa \tau_x}{\tau_x +\tau_y}\right] x_i + \left[\frac{\tau_y \alpha}{\tau_x +\tau_y} +(1-\alpha) \frac{\kappa \tau_y}{\tau_x+\tau_y}+(1-\alpha)(1-\kappa)\right]y\\&=
	\left[\frac{ \tau_x (\alpha  +(1-\alpha)\kappa)}{\tau_x +\tau_y} \right] x_i + \left[(1-\alpha)(1-\kappa) + \frac{\alpha \tau_y +(1-\alpha) \kappa \tau_y}{\tau_x +\tau_y}\right]y
	\end{align*}
	Matching coefficients and solving:
	\begin{align*}
	\left[(\alpha + (1-\alpha)\kappa)\frac{\tau_x}{\tau_x+\tau_y}\right] = \kappa \implies \kappa =\frac{\alpha\tau_x}{\alpha \tau_x +\tau_y}
	\end{align*}
	as desired. 
\end{proof}

\begin{remark*}
	As we take the limit as $n$ goes to infinity, we recover $\kappa=\frac{\alpha\tau_x}{\alpha\tau_x +\tau_y}$ as in the continuum setting and in \citet{morrisshin}.
\end{remark*}

\begin{lemma*}[\ref{lem:finite-eukbp}]

	The expected utility of agents for playing $\theta_i=\kappa x_i+(1-\kappa)y$, given $s$ is:
	\begin{align*}
	\E_i[u_i\vert s]= -\alpha (\kappa^2 \sigma^2_x + (1-\kappa)^2\sigma^2_y) -\frac{(1-\alpha)\kappa^2(n-1)}{n}\left.\sigma^2_x \right.
	\end{align*}
\end{lemma*}
\begin{proof}
	We can write the expected utility of agents as
	\begin{align*}
	\E_i[u_i\vert s] = -\alpha \E_i[(\theta_i-s)^2\vert s] -(1-\alpha)\E_i\left[\left.\left(\theta_i-\bar{\theta}\right)^2\right\vert s\right]
	\end{align*}
	
	We can plug in the equilibrium strategies and recall the definition of $\bar{\theta}$ to write
	
	\begin{align*}
	\E_i[u_i\vert s]&= -\alpha \E_i[(\theta_i-s)^2\vert s]  -(1-\alpha)\E_i\left[\left.\left(\theta_i-\frac{1}{n}\sum\limits_{j=1}^n \theta_j\right)^2\right\vert s\right]\\
	\end{align*}
	
	We can write each $\theta_j=\kappa x_j + (1-\kappa)y$ and since we are taking expectations conditional on knowing $s$, we can write the signals $x_j$ and $y$ as the state $s$ plus mean-zero Gaussian noise with variance $\sigma^2_x$ and $\sigma^2_y$, respectively.  We write these as $s+\varepsilon_{x_j}$ and $s+\varepsilon_y$, so we can write $\theta_j=s+\kappa \varepsilon_{x_j} + (1-\kappa)\varepsilon_y$.
	
	\begin{align*}
	\E_i[u_i\vert s]&= -\alpha \E_i[(s+\kappa \varepsilon_{x_i} + (1-\kappa)\varepsilon_y-s)^2\vert s] \\ &\qquad-(1-\alpha)\E_i\left[\left.\left(s+\kappa \varepsilon_{x_i} + (1-\kappa)\varepsilon_y-\frac{1}{n}\sum\limits_{j=1}^n s+\kappa \varepsilon_{x_j} + (1-\kappa)\varepsilon_y\right)^2\right\vert s\right]\\
	&= -\alpha \E_i[(\kappa \varepsilon_{x_i} + (1-\kappa)\varepsilon_y)^2] -(1-\alpha)\E_i\left[\left.\left(\kappa \varepsilon_{x_i} + -\frac{1}{n}\sum\limits_{j=1}^n \kappa \varepsilon_{x_j} \right)^2\right.\right]\\
	&= -\alpha \E_i[(\kappa \varepsilon_{x_i} + (1-\kappa)\varepsilon_y)^2] -(1-\alpha)\kappa^2\E_i\left[\left.\left( \varepsilon_{x_i} + -\frac{1}{n}\sum\limits_{j=1}^n \varepsilon_{x_j} \right)^2\right.\right]\\
	&= -\alpha \E_i[(\kappa \varepsilon_{x_i} + (1-\kappa)\varepsilon_y)^2] -(1-\alpha)\kappa^2\E_i\left[\left.\left(\frac{n-1}{n} \varepsilon_{x_i} + -\frac{1}{n}\sum\limits_{j\neq i}^n \varepsilon_{x_j} \right)^2\right.\right]\\
	&= -\alpha \E_i[(\kappa \varepsilon_{x_i} + (1-\kappa)\varepsilon_y)^2] -\frac{(1-\alpha)\kappa^2}{n^2}\E_i\left[\left.\left((n-1) \varepsilon_{x_i} + -\sum\limits_{j\neq i}^n \varepsilon_{x_j} \right)^2\right.\right]\\
	\end{align*}
	
	Because all of the $\varepsilon_{x_j}$ and $\varepsilon_y$ are independent with mean zero, we can write
	
	\begin{align*}
	\E_i[u_i\vert s] &= -\alpha (\kappa^2 \sigma^2_x + (1-\kappa)^2\sigma^2_y) -\frac{(1-\alpha)\kappa^2}{n^2}\left.\left( (n-1)^2\sigma^2_x + (n-1)\sigma^2_x \right)^2\right.\\
	&= -\alpha (\kappa^2 \sigma^2_x + (1-\kappa)^2\sigma^2_y) -\frac{(1-\alpha)\kappa^2(n-1)}{n}\left.\sigma^2_x \right.\\
	\end{align*}

\end{proof}

\begin{proposition*}[\ref{lem:eukbp}]

	The expected utility of agents for playing $\theta_i=\kappa x_i+(1-\kappa)y$, conditional on $s$,  is:
	\begin{align*}
	\E_i[u_i\vert s] = -\alpha(1-\kappa)^2\sigma^2_y - \kappa^2\sigma^2_x
	\end{align*}
\end{proposition*}

\begin{proof}
	We can write the expected utility of agents as
	\begin{align*}
	\E_i[u_i\vert s] = -\alpha \E_i[(\theta_i-s)^2\vert s] -(1-\alpha)\E_i\left[\left.\left(\theta_i-\bar{\theta}\right)^2\right\vert s\right]
	\end{align*}
	
	We can plug in the equilibrium strategies and recall the definition of $\bar{\theta}$ to write
	
	\begin{align*}
	\E_i[u_i\vert s]&= -\alpha \E_i[(\theta_i-s)^2\vert s]  -(1-\alpha)\E_i\left[\left.\left(\theta_i- \int_0^1 \theta_j dj\right)^2\right\vert s\right]\\
	\end{align*}
	
	We can write each $\theta_j=\kappa x_j + (1-\kappa)y$ and since we are taking expectations conditional on knowing $s$, we can write the signals $x_j$ and $y$ as the state $s$ plus mean-zero Gaussian noise with variance $\sigma^2_x$ and $\sigma^2_y$, respectively.  We write these as $s+\varepsilon_{x_j}$ and $s+\varepsilon_y$, so we can write $\theta_j=s+\kappa \varepsilon_{x_j} + (1-\kappa)\varepsilon_y$.
	
	\begin{align*}
	\E_i[u_i\vert s]&= -\alpha \E_i[(s+\kappa \varepsilon_{x_i} + (1-\kappa)\varepsilon_y-s)^2\vert s] \\ &\qquad-(1-\alpha)\E_i\left[\left.\left(s+\kappa \varepsilon_{x_i} + (1-\kappa)\varepsilon_y- \int_0^1 \left(s+\kappa \varepsilon_{x_j} + (1-\kappa)\varepsilon_y dj\right)\right)^2\right\vert s\right]\\
	&= -\alpha \E_i[(\kappa \varepsilon_{x_i} + (1-\kappa)\varepsilon_y)^2] -(1-\alpha)\E_i\left[\left.\left(\kappa \varepsilon_{x_i} + \int_0^1 \left(\kappa \varepsilon_{x_j} dj\right)\right)^2\right.\right]\\
	\end{align*}
	
	Because $\E[\varepsilon_{x_j}]=0$, $\int_0^1 \left(\kappa \varepsilon_{x_j} dj\right)=0$.  Furthermore, since all of the $\varepsilon_{x_j}$ and $\varepsilon_y$ are independent we have
	
	\begin{align*}
	\E_i[u_i\vert s]&= -\alpha \E_i[(\kappa \varepsilon_{x_i} + (1-\kappa)\varepsilon_y)^2] -(1-\alpha)\E_i\left[\left.\left(\kappa \varepsilon_{x_i} \right)^2\right.\right]\\
	&= -\alpha (\kappa^2\sigma^2_x + (1-\kappa)^2\sigma^2_y) -(1-\alpha)\kappa^2\sigma^2_x\\
	&= -\alpha(1-\kappa)^2\sigma^2_y - \kappa^2\sigma^2_x
	\end{align*}

\end{proof}

\subsection*{Omitted Proofs from \Cref{sec:extend}}

\begin{claim*}[\ref{clm:meanzero}]
	
	If  there exists an equilibrium in noisy strategies, where player $i$ chooses $\noisy_i = \theta_i+ \eta_i$ and each player's $\eta_i$ is drawn independently from a distribution $H_i$, then in particular there exists such an equilibrium strategy profile in which the mean of each $H_i$ is zero.
	
\end{claim*}

\begin{proof}
		Since the distributions which generate the $\eta_i$ are revealed after each player announces $\noisy_i$, choosing to draw noise from a distribution with non-zero mean cannot improve the privacy that player $i$ achieves, since a representative player $j$ can simply subtract this mean when constructing her posterior distribution over $x_i$.  Additionally, choosing a mean other than zero makes the utility from the guessing component strictly worse.  Finally, if we assume that all players other than $i$ choose to add noise drawn from a distribution with the same mean, then in the coordination portion of the utility function, player $i$ choosing a noise distribution mean other than the common one is dominated by choosing the common one.  In particular, if everyone else chooses mean-zero noise, player $i$ should as well.

\end{proof}

\begin{lemma*}[\ref{lem:finite-sep}]
	
	In the game with finitely many players, the players' utility functions in the privacy-aware game  separate additively into the utility in the privacy unaware game, a penalty in $\nu_i$, and a privacy term as
	
	\begin{align*}
	v_i(\noisy_i, \noisy_{-i}) &= (1-\beta)(-\alpha(\noisy_i -s)^2 -(1-\alpha)(\noisy_i-\bar{\noisy})^2 +\beta\rho(\noisy_i) \\
	&= (1-\beta)u_i(\theta_i, \noisy_{-i})  +(1-\beta)\left(\alpha +\left(1-\frac{1}{n}\right)^2 (1-\alpha)\right) \nu_i +\beta \rho(\noisy_i),
	\end{align*}
	where $\nu_i$ denotes the variance of the noise-generating distribution $H_i$ of player $i$ and $u_i$ is the utility function in the privacy-unaware game.
\end{lemma*}
\begin{proof}
	 Writing $\noisy_i$ as $\theta_i+\eta_i$, i.e. a deterministic component plus random noise, we can decompose the various pieces of the utility function as follows.  The first part is 
	\begin{align*}
	-\alpha \E_i[(\theta_i +\eta_i-s)^2] = -\alpha\E_i\left[(\theta_i-s)^2+\eta_i(\theta_i-s)+\eta_i^2\right] = - \alpha \E_i[(\theta_i-s)^2] - \alpha \nu_i ,
	\end{align*}
	where we have again used the independence of $\eta_i$ to conclude that $E_i[\eta_i (\theta_i-s)]=0$. 
	
	The second term is 
	\begin{align*}
	-(1-\alpha) \E_i\left[\left(\theta_i -\eta_i -\frac{1}{n} \sum\limits_{j=1}^n \noisy_j\right)^2\right] = -(1-\alpha) \E_i \left[\left(\theta_i\left(1-\frac{1}{n}\right) + \eta_i\left(1-\frac{1}{n}\right) -\frac{1}{n} \sum\limits_{j\neq i}\noisy_j\right)^2\right]
	\end{align*}
	Rewriting gives: 
	\begin{align*}-(1-\alpha)\E_i \left[\left(\theta_i\left(1-\frac{1}{n}\right) - \frac{1}{n} \sum_{j\neq i} \noisy_j\right)^2 \right] -(1-\alpha) \left(1-\frac{1}{n}\right) \nu_i -(1-\alpha)\E_i\left[\eta_i\left(1-\frac{1}{n}\right)^2\frac{1}{n}\theta_i\sum_{j\neq i} \noisy_j\right]
	\end{align*}
	After collecting terms, we see that it suffices to argue that the final term
	\begin{align*}
	(1-\alpha)\E_i\left[\eta_i\left(1-\frac{1}{n}\right)^2\frac{1}{n}\theta_i\sum_{j\neq i} \noisy_j\right] =0,
	\end{align*}
	which follows from the fact that $\eta_i$ is independent of all the other parameters of the game as well as the other $\eta_j$.
	
\end{proof}

\begin{corollary*}[\ref{lem:sep}]

	The proof is nearly identical to that of \Cref{lem:finite-sep}. Suppose that all agents play a noisy strategy $\noisy_i = \theta_i+ \eta_i$, with $\eta_i$ being a random variable and $\mathbb{E}[\eta_i]=0$. Then an agent's utility can be decomposed into
	\begin{align*}
	\E_i [v_i(\noisy,\noisy_{-i})] =  (1-\beta) \E_i [u(\theta_i, \noisy_{-i})] + (1-\beta) \nu_i + \beta \rho(\noisy_i)
	\end{align*}
	where $u_i$ is the utility function in the original game and $\nu_i$ is the variance of the noise-generating distribution $H_i$. 
\end{corollary*}
\begin{proof}
	By definition,
	\begin{align*}
	\E_i [v_i(\noisy,\noisy_{-i})] = (1-\beta) \E_i [u_i(\noisy_i,\noisy_{-i})] + \beta \rho(\noisy_i),
	\end{align*}
	so if we show that $u_i(\noisy_i, \noisy_{-i}) = u_i(\theta_i,\noisy_{-i}) - \nu_i$ we will be done. 
	We can write 
	\begin{align*}
	\E_i [u_i(\noisy_i,\noisy_{-i})] &= -\alpha \E_i[(\noisy_i - s)^2]   -(1-\alpha)  \E_i[(\noisy_i - \bar{\noisy})^2] \\ 
	&= -\alpha \E_i [(\theta_i + \eta_i -s)^2 ] -(1-\alpha) \E_i[(\theta_i + \eta_i -\bar{\noisy})^2] \\
	&= -\alpha \E_i[(\theta_i - s + \eta_i)^2] -(1-\alpha)\E_i[(\theta_i -\bar{\noisy} + \eta_i)^2] \end{align*}
	Expanding these terms, we have
	\begin{align*}
	\E_i [u_i(\noisy_i,\noisy_{-i})] &= -\alpha \left(\E_i[(\theta_i-s)^2] +2 \E_i[\eta_i(\theta_i-s)] + \E_i[(\eta_i^2)]\right)  \\ &\qquad  -(1-\alpha) \left(\E_i[(\theta_i-\bar{\noisy})^2] + 2 \E_i[(\eta_i)(\theta_i-\bar{\noisy})] + \E_i[\eta_i^2]\right)
	\end{align*}
	
	Now the first terms of each line sum to exactly $u_i(\theta_i,\noisy_{-i})$. On the other hand, the sum of the last two terms is $-\E_i[\eta_i^2]=-\nu_i$. To complete the proof, we show that these middle to terms are, in fact, zero. To see this, notice that at $\info_i$, $\eta_i$ is yet unrealized with $\E_i[\eta_i]=0$, but is independent of $s$ and $\theta_i$ and thus of $\bar{\theta}$. Hence,
	\begin{align*}
	\E_i[\eta_i(\theta_i-s)] =\E_i[\eta_i] \E_i[\theta_i-s] = 0,
	\end{align*}
	Moreover, $\eta_i$ is independent of each ${\noisy}_{-i}$, and agent $i$'s action cannot unilaterally change $\bar{\noisy}$, so
	\begin{align*}
	\E_i[\eta_i (\theta_i -\bar{\noisy})]=\E_i[\eta_i (\theta_i - \bar{\noisy}_{-i})]=\E_i[\eta_i]\E_i[\theta_i-\bar{\noisy}_{-i}]=0
	\end{align*}
	
\end{proof}

\begin{lemma*}[\ref{lem:finite-ext-foc}]
	
	In an equilibrium of the game with finitely many players where the optimal action is $\noisy^*_i = \theta^*_i + \eta_i$,  the optimal choice of $\theta^*_i$ and the variance $\nu^*$ for the noise-generating distribution $H_i$ from which $\eta_i$ is drawn must satisfy
	
	\begin{align*}
	\theta_i^* = \frac{\alpha n^2 \E_i[s]}{\alpha (2n-1)+(n-1)^2} + \frac{(1-\alpha)(n-1)\E_i\left[\sum\limits_{j\neq i}\theta_{j}\right]}{\alpha (2n-1)+(n-1)^2} \\ 	
	\end{align*} 
	and 
	\begin{align*} 
	\frac{\partial \rho}{\partial \nu^*} =-\frac{-(1-\beta)\left(\alpha +\left(1-\frac{1}{n}\right)^2 (1-\alpha)\right)}{\beta}.
	\end{align*}

\end{lemma*}

\begin{proof}
	Using \Cref{lem:finite-sep}, we can decompose the expected utility of agent $i$ as
	\begin{align*}
	\E_i [v_i(\noisy_i, \noisy_{-i} ] = (1-\beta)\E_i[u_i(\theta_i, \noisy_{-i})]  -(1-\beta)\left(\alpha +\left(1-\frac{1}{n}\right)^2 (1-\alpha)\right) \nu_i +\beta \rho(\noisy_i),
	\end{align*}
	
	which is the sum of a piece that depends on $\theta_i$ and a piece that depends on $\nu_i$.
	
	The agent can therefore optimize each piece separately with her choice of $\theta_i$ and $\nu_i$. \Cref{lem:finite-foc} { again} gives the first order condition on $\theta_i^*$ - { the optimal deterministic component is a convex combination of expectations about the state and the average action (which, under the assumption of a noisy symmetric linear equilibrium, contains added, but mean-zero, noise). 
	}
	
	To find the first order condition on $\nu^*$, we can write
	
	\begin{align*}
	0=\frac{\partial v_i}{\partial \nu^*} = -(1-\beta)\left(\alpha +\left(1-\frac{1}{n}\right)^2 (1-\alpha)\right) +\beta\frac{\partial\rho}{\partial \nu^*}
	\end{align*}
	
	and solve for $\frac{\partial\rho}{\partial\nu^*}$ to get the result.

\end{proof}

\begin{proposition*}[\ref{lem:FOCext}]
	In an equilibrium where the optimal action is $\noisy^*_i = \theta^*_i + \eta_i$,  the optimal choice of $\theta^*_i$ and the variance $\nu^*$ for the noise-generating distribution $H_i$ from which $\eta_i$ is drawn must satisfy
	
	\begin{align*}
	\theta_i^* = \alpha \E_i[s] +(1-\alpha) \E_i[\bar{\noisy}] \qquad 	\frac{\partial \rho}{\partial \nu^*} = -\frac{1-\beta}{\beta}.
	\end{align*}
	
\end{proposition*}
\begin{proof}
	Using Lemma \ref{lem:sep}, we can decompose the utility of agent $i$ as
	\begin{align*}
	\E_i [v_i(\noisy,\noisy_{-i})] =  (1-\beta) \E_i [u(\theta_i, \noisy_{-i})] - (1-\beta) \nu_i - \beta \rho(\noisy_i),
	\end{align*}
	
	which is the sum of a piece that depends on $\theta_i$ and a piece that depends on $\nu_i$.
	
	The agent can therefore optimize each piece separately with her choice of $\theta_i$ and $\nu_i$. \Cref{lem:FOC} gives the first order condition on $\theta_i^*$.
	
	To find the first order condition on $\nu^*$, we can write
	
	\begin{align*}
	0=\frac{\partial v_i}{\partial \nu^*} = -(1-\beta) -\beta\frac{\partial\rho}{\partial \nu^*}
	\end{align*}
	
	and solve for $\frac{\partial\rho}{\partial\nu^*}$ to get the result.

\end{proof}

\mute{
\begin{corollary*}[\ref{cor:finiteparams}]
	
	The optimal linear $\theta_i$ is, as before, 
	\begin{align*}
	\theta_i^* =  \kappa x_i + (1-\kappa) y
	\end{align*}
	and the optimal choice of variance for the noise distribution is
	\begin{align*}
	\nu_{i,prec}^* = \sqrt{\frac{\beta}{1-\beta} \left(\alpha + (1-\alpha)\left(1-\frac{1}{n}\right)^2\right)} \end{align*}and \begin{align*} \nu_{i,ent}^* = {\frac{\beta}{1-\beta} \left(\alpha + (1-\alpha)\left(1-\frac{1}{n}\right)^2\right)} 
	\end{align*}
	where $\nu_{i,prec}^*$ and $\nu_{i,ent}^*$ are the optimal variances under $\rho$ being the precision and entropy privacy measures, respectively.  
	
	The proof of this fact is identical to that of \Cref{cor:params} in the infinite game.
\end{corollary*}

}

\subsection*{Omitted Proofs from \Cref{sec:pop}}

\begin{lemma*}[\ref{lem:popform}]
	The price of privacy in the game where agents play a linear strategy has the form
	\begin{align*} 
	\PoP(\tau_x,\tau_y,\beta) =1 +\frac{\nu^*_i}{ \E_i [u_i(\theta_i^*,\theta_{-i}^*)]}
	\end{align*}
	where the expected utility is with respect to the game with signal variances $\sigma^2_x$ and $\sigma^2_y$.
\end{lemma*}
\begin{proof}
		First note that 
		\begin{align*}
		\E_i[u_i(\theta_i, \noisy_{-i})]= \E_i[u_i(\theta_i, \theta_{-i})] 
		\end{align*}
		because the noise added to $\theta_i$ has a mean of zero.  Now, the\Cref{lem:sep} lets us write
		\begin{align*}
		\PoP(\tau_x,\tau_y,\beta) = \frac{\E_i [u_i(\noisy_i^*,\noisy_{-i}^*)]}{\E_i [u_i(\theta_i^*,\theta_{-i}^*)]}  = \frac{\E_i[ u_i(\noisy_i^*,\noisy_{-i}^*)] + \nu^*_i }{\E_i[ u_i(\theta_i^*,\theta_{-i}^*)]}
		\end{align*}
		where we have factored out all of the negative signs.  Combining with the previous part, we have that 
		\begin{align*}
		\PoP(\tau_x,\tau_y,\beta) = \frac{\E_i[ u_i(\noisy_i^*,\noisy_{-i}^*)] + \nu^*_i }{\E_i[ u_i(\theta_i^*,\theta_{-i}^*)]} = \frac{\E_i[ u_i(\noisy_i^*,\theta{-i}^*)] + \nu^*_i }{\E_i[ u_i(\theta_i^*,\theta_{-i}^*)]} = 1 + \frac{\nu^*_i}{\E_i [u_i(\theta_i^*,\theta_{-i}^*)]},
		\end{align*}
		as desired.
\end{proof}

\begin{lemma*}[\ref{lem:util-agg}]
	
	Consider an instance of the privacy-aware game where an aggregator observes the actions of $n$ agents (either all all of the agents in the finite case or some uniformly random sample in the finite or infinite case), the  signal variances are $\sigma^2_x$ and $\sigma^2_y$, and players choose to add mean-zero noise with variance $\nu^*_i$.  Then the utility of the aggregator, as measured by the variance of the sample average about the true state $s$ is given by
	\begin{align*}
	\mathcal{U}_{agg}(\sigma^2_x,\sigma^2,\nu^*_i,n) &=\E\left[  \left(\left.\left(\frac{1}{n}\sum\limits_{i=1}^n \noisy_i\right)  - s \right)^2       \right\vert s \right]\\
	&=\frac{\kappa^2}{n}\sigma^2_x + \frac{\nu^*_i}{n} + (1-\kappa)^2\sigma^2_y
	\end{align*}
	
	We can find the utility in the privacy-unaware game by letting $\nu^*_i=0$.  This recovers a result in \citet{morrisshin}.
\end{lemma*}

\begin{proof}

		\begin{align*}
		\mathcal{U}_{agg}(\sigma^2_x,\nu_i^*,\sigma^2,n) &=\E\left[  \left(\left.\left(\frac{1}{n}\sum\limits_{i=1}^n \noisy_i\right)  - s \right)^2       \right\vert s \right]\\
		&=\E\left[  \left(\left.\left(\frac{1}{n}\sum\limits_{i=1}^n \kappa x_i + (1-\kappa)y + \eta_i \right)  - s \right)^2       \right\vert s \right]\\
		&=\E\left[  \left(\left.\left(\frac{1}{n}\sum\limits_{i=1}^n \kappa x_i + \eta_i \right) + (1-\kappa)y  - s \right)^2       \right\vert s \right]\\
		&=\E\left[  \left(\left.\left(\frac{1}{n}\sum\limits_{i=1}^n \kappa (s + \varepsilon_{x_i} ) +\eta_i \right) + (1-\kappa)(s+\varepsilon_y)  - s \right)^2       \right\vert s \right]\\
		&=\E\left[  \left(\left.\left(\frac{1}{n}\sum\limits_{i=1}^n \kappa\varepsilon_{x_i} + \eta_i \right) + (1-\kappa)\varepsilon_y  \right)^2       \right\vert s \right]\\
		&=\E\left[  \left(\frac{\kappa}{n}\sum\limits_{i=1}^n \varepsilon_{x_i} + \frac{1}{n} \sum\limits_{i=1}^{n} \eta_i  + (1-\kappa)\varepsilon_y   \right)^2        \right]\\
		&=\frac{\kappa^2}{n}\sigma^2_x + \frac{\nu^*_i}{n} + (1-\kappa)^2\sigma^2_y
		\end{align*}
		
		where we have decomposed $x_i$ and $y$ into $s$ plus mean-zero Gaussian noise $\varepsilon_{x_i}$ and $\varepsilon_y$.
\end{proof}

\fi
\end{document}